\documentclass[journal]{IEEEtran}
\ifCLASSINFOpdf
\else
\fi 
\usepackage[tbtags]{amsmath}
\usepackage{amsthm,amsmath,amssymb,amsfonts}
\usepackage{graphicx}
\graphicspath{{pic/}}
\usepackage{subfigure}
\usepackage{multirow}
\usepackage[ruled]{algorithm2e}
\usepackage{setspace}
\usepackage{mathrsfs}

\usepackage{booktabs}
\usepackage{multicol}
\usepackage{arydshln}
\usepackage{stfloats}
\usepackage[misc]{ifsym}
\usepackage{algpseudocode}
\algrenewcommand\textproc{\texttt}
\usepackage[square,comma,sort&compress,numbers]{natbib}

\newtheorem{Theorem}{Theorem}

\newtheorem{Definition}{Definition}
\newtheorem{Property}{Property}
\newtheorem{Corollary}{Corollary}

\usepackage{comment}
\usepackage{soul}	
\newcommand\Lcomment[1]{\textcolor{black}{#1}} 
\usepackage[dvipsnames,usenames]{color}
\definecolor{dgreen}{rgb}{0,.6,0}
\newcommand\modified[1]{\textcolor{black}{#1}} 

\usepackage[dvipsnames,usenames]{color}
\usepackage{etoolbox}
\let\mybibitem\bibitem
\renewcommand{\bibitem}[1]{%
  \ifboolexpr{ 
  test {\ifstrequal{#1}{alawida2019deterministic}}
    or
  test {\ifstrequal{#1}{alawida2020enhanced}}
    or
 test {\ifstrequal{#1}{alawida2021novel}}
    or
 test {\ifstrequal{#1}{alawida2019digital}}
    or
 test {\ifstrequal{#1}{shujun2001pseudo}}
 or
 test {\ifstrequal{#1}{hua2017designing}}
 or
 test {\ifstrequal{#1}{lopez2016mackey}}
  }
 {\color{black}\mybibitem{#1}}
    {\color{black}\mybibitem{#1}}%
}

\begin{document}
\title{From Chaos to Pseudo-Randomness: A Case Study on the 2D Coupled Map Lattice}
\author{
Yong~Wang,
        ~Zhuo~Liu,
        ~Leo~Yu~Zhang,~\IEEEmembership{Member,~IEEE,}
        ~Fabio Pareschi,~\IEEEmembership{Senior Member,~IEEE,}\\
        Gianluca Setti,~\IEEEmembership{Fellow,~IEEE}
        and~Guanrong Chen,~\IEEEmembership{Life Fellow,~IEEE}
\thanks{Y. Wang and Z. Liu are with the College of Computer Science and Technology, Chongqing University of Posts and Telecommunications, China (email: wangyong1@cqupt.edu.cn; 
liuzhuo1987@outlook.com).}
\thanks{L. Zhang is with the School of Information Technology, Deakin
University, Australia (email: leo.zhang@deakin.edu.au).}
\thanks{F. Pareschi and G. Setti are with the Department of Electronics and Telecommunications, Politecnico di Torino, Italy (email: fabio.pareschi@polito.it; gianluca.setti@polito.it).}
\thanks{G. Chen is with the Department of Electrical Engineering, City University of Hong Kong, Hong Kong SAR (email: eegchen@cityu.edu.hk).}
}

\maketitle

\begin{abstract}
Applying chaos theory for secure digital communications is promising and it is well acknowledged that in such applications the underlying chaotic systems should be carefully chosen. However, the requirements imposed on the chaotic systems are usually heuristic, without theoretic guarantee for the resultant communication scheme. Among all the primitives for secure communications, it is well-accepted that (pseudo) random numbers are most essential. 
Taking the well-studied two-dimensional coupled map lattice ($2$D CML) as an example, this paper performs a theoretical study towards pseudo-random number generation with the $2$D CML.
In so doing, an analytical expression of the Lyapunov exponent (LE) spectrum of the $2$D CML is first derived. Using the LEs, one can configure system parameters to ensure the $2$D CML only exhibits complex dynamic behavior, and then collect pseudo-random numbers from the system orbits.
Moreover, based on the observation that least significant bit distributes more evenly in the (pseudo) random distribution, an extraction algorithm $\mathbf{E}$ is developed with the property that, when applied to the orbits of the $2$D CML, it can squeeze uniform bits. 
In implementation, if fixed-point arithmetic is used in binary format with a precision of $z$ bits after the radix point, $\mathbf{E}$ can ensure that the deviation of the squeezed bits is bounded by $2^{-z}$. 
\modified{Further simulation results demonstrate that the new method not only guide the $2$D CML model to exhibit complex dynamic behavior, but also generate uniformly distributed independent bits with good efficiency}. 
In particular, the squeezed pseudo-random bits can pass both NIST 800-22 and TestU01 test suites in various settings.
This study thereby provides a theoretical basis for effectively applying the $2$D CML to secure communications.

\end{abstract}
\begin{IEEEkeywords}
Chaos; Lyapunov Exponent; Random Number Generator; Secure Communication; $2$D Coupled Map Lattice.
\end{IEEEkeywords}
\IEEEpeerreviewmaketitle

\section{Introduction}
\label{Sec:intro}
\IEEEPARstart
{I}{n the} past two decades, chaotic systems have been widely used for secure communications, \modified{owing to their unique qualities including} complexity, pseudo-randomness and ergodicity \cite{jakimoski2001chaos,yang1997cryptography}. 
\modified{
These traits are related to confusion and diffusion, which are desired characteristics for good cryptographic primitives, according to \cite{shannon1949communication}.
For this reason, many encryption algorithms and secure communication schemes based on chaotic systems have been proposed,} creating a great deal of research across nonlinear dynamics and information security \cite{kocarev2001chaos,bergamo2005security,el2014chaos,zhang2017security}. %
Essentially, \modified{chaotic systems' inherent features play a significant part} in maintaining the schemes' security.
\modified{As a result}, selecting or constructing a chaotic system from the perspective of cryptographic applications has become \modified{a critical challenge for chaos-based secure communications}.

\modified{There are two different options on how to choose chaotic systems for secure communications \cite{alvarez2004breaking}. The first one is to choose simple chaotic systems since simple chaotic systems generally have fewer arithmetic operations, and thus less computational complexity. 
Secure communication schemes based on simple systems typically have better runtime efficiency \cite{hua2017one,zhang2017security}. 
However, if the structure of a simple chaotic system is not complicated enough, an attacker might be able to predict its chaotic orbits \cite{lopez2016mackey}, which are used for secure communication directly or indirectly \cite{cao2015high,miranian2012developing}.}
This probably brings some potential vulnerability to the secure communication schemes.

To alleviate this problem, on the one hand, \modified{one may consider enhancing the dynamic complexity of simple chaotic systems \cite{wang2016pseudorandom,zhou2014cascade,hua2017designing,hua2015dynamic,alawida2019deterministic,alawida2020enhanced,alawida2021novel,alawida2019digital}.
In this way, a better trade-off between security and efficiency can be obtained if the enhancing method is still efficient.}
On the other hand, one may use higher-dimensional chaotic systems to design better secure communication schemes \cite{li2018novel,dachselt2001chaos}. 
Compared with simple chaotic systems, \modified{such as the logistic and the tent maps}, a higher-dimensional chaotic system typically has more complicated dynamic behavior. \modified{Generally, it is more difficult, if not impossible, to predict the sequences generated from higher-dimensional chaotic systems.}

However, \modified{this gain in security is not free, since more arithmetic operations of a higher-dimensional chaotic system will consume more computational capacity and incur inferior running speed.}
To decrease the number of iterations of a higher-dimensional system and achieve higher efficiency, it is common to extract more pseudo-random bits from a single iteration of the \modified{underlying} higher-dimensional chaotic system \cite{li2006multiple,wang2016pseudorandom,lv2018novel}. {But this is heuristic and depends on the pseudo-randomness of the underlying chaotic system}. 
\modified{Another method is to iterate the higher-dimensional chaotic system in parallel to improve the efficiency. However, it only applies to higher-dimensional chaotic systems that support parallelization implementation.}

Besides the consideration about the trade-off between security and efficiency, another critical issue is the effect of finite-precision representation of chaotic orbits \cite{wang2016theoretical}. 
\modified{Generally, chaotic orbits are real numbers, and real numbers are then truncated and represented in either floating-point or  fixed-point arithmetic on digital computer}. 
\modified{When expressed with a certain precision under the fixed-point arithmetic, all chaotic systems will inevitably degenerate to  become periodic.}
Worse yet, if chaotic systems are not configured correctly, their orbits, even represented by real numbers, could fall into periodic with a short \modified{period}.
Without deliberate design, therefore, these problems will degrade the security of chaos-based encryption algorithms. 
\modified{Based on these observations}, to design a fully-fledged chaotic secure communication system, it is \modified{imperative and indispensable to have some theoretic guidelines for ensuring chaotic orbits to run in full chaotic state and the pseudo-random bits extracted from digitized orbits} are evenly distributed with a sufficiently long \modified{period}.
The present work tries to address this issue by taking the two-dimensional coupled map lattice (2D CML) for a case study.

The coupled map lattice \cite{kaneko1989pattern} is a classic model of spatiotemporal chaos. 
It is a complex two-directional chaotic system coupled with multiple identical simple chaotic maps, and it has excellent scalability. 
\modified{All the nodes in the CML have the same structure, so their arithmetic operations can be computed individually, which makes it  possible to run the CML in parallel by certain special design \cite{wang2011parallel}.}
Moreover, since the CML consists of multiple nodes, it will degenerate to a periodic system only \modified{if} all the nodes are in periodic state \modified{simultaneously}. \modified{From this view, even} under finite precision representation, the period of CML model is longer than that of a single node (a simple chaotic map) \cite{coombes2000period}, which can easily make the orbital period to be long enough for practical applications.

\modified{CML's potential has been utilised in recent years to create numerous secure communication algorithms \cite{kumar2018cryptographic,lv2018novel,wong2008fast,fu2018image,hussain2015image,wang2008one,zhang2014symmetric,newman2004finding} due to the qualities described above.}
For example, by using the CML as the core of diffusion operations, a cryptographic model is proposed to guarantee better security of information processes \cite{kumar2018cryptographic}. 
\modified{The work \cite{lv2018novel} proposes a novel CML-based pseudo-random number generator (PRNG) with strong potential for cryptographic applications.
The sequences generated from CML are used for constructing nonlinear substitution boxes (S-boxes) \cite{wong2008fast}, which is a basic building block for many encryption algorithms.
In \cite{fu2018image}, the way of designing S-boxes based on CML is assembled to an image encryption algorithm to enhance the image privacy.
CML can also be used with other technologies, such as DNA coding, to create more efficient encryption methods \cite{hussain2015image}.}
From a theoretical perspective, to further enhance the complexity of dynamic behavior, the one-dimensional CML is extended to the two-dimensional form \cite{wang2008one}. Some hash functions and encryption algorithms are \modified{heuristically} developed based on the complexity, diffusion and randomness of the $2$D CML \cite{zhang2014symmetric,newman2004finding}.

\modified{When different values for the CML parameters are specified, the state of a CML system (either $1$D or $2$D) can exhibit diverse patterns, such as frozen random pattern, competition intermittency, and fully developed chaos, according to \cite{kaneko1989pattern,biswas2016patterns}. To our knowledge, however, there is little theoretical study on the dynamics of $2$D CML, despite the fact that the $2$D CML has previously been heuristically used to secure communications \cite{wang2008one,zhang2014symmetric,newman2004finding}. Furthermore, as previously stated,} heuristically extracting pseudo-random bits from chaotic orbits would leave a security flaw for the integrated system, which is also the case with the $2$D CML.


To address these challenges, this paper studies the dynamics of $2$D CML as well as pseudo-random bit generation from orbits of $2$D CML in a theoretical perspective. 
This work makes the following contributions: 
\begin{itemize}
    \item For a commonly used $2$D CML system, the analytic formula between the Lyapunov exponents (LEs) and the parameters of the model is derived, which provides theoretical guarantee for ensuring the $2$D CML to run in fully developed chaotic state. 
   \modified{
    \item Focusing on the fixed-point arithmetic, an extraction algorithm $\mathbf{E}$ is introduced to process the digitized orbits of the $2$D CML, which ensures the pseudo-random bits extracted from the orbits to be uniform under certain mild assumption. This extraction algorithm, concerning the acquisition of uniform bits from non-uniform random sources, is also of research interest in its own right. 
    \item Extensive experiments are performed to verify the theoretical results. In particular, it is demonstrated that the extracted bits from the orbits of $2$D CML successfully pass both NIST 800-22 and TestU01 tests \cite{bassham2010sp,l2007testu01} under different settings. Moreover, efficiency comparison validates that the proposed PRNG is faster than the only known theoretic sound chaotic PRNG \cite{shujun2001pseudo}, and also faster than some recently designed heuristic PRNGs \cite{alawida2020enhanced,lv2018novel}. }
\end{itemize}

The rest of the paper is organized as follows. Some preliminary knowledge is given in Sec.~\ref{sec:Pre}. In Sec.~\ref{sec:method}, theoretical Lyapunov exponent analysis of the $2$D CML is presented and the practical question of how to extract uniform random bit from the orbits of the $2$D CML is addressed. 
In Sec.~\ref{sec:SimTest}, some numerical tests are presented to verify the theoretical results. Finally, conclusion is drawn in Sec.~\ref{sec:conclusion}.

\section{Preliminaries}
\label{sec:Pre}
To prepare for the technical development of this work, some preliminary knowledge is provided in this section. 

\begin{figure}[ht]
	\centering
	\includegraphics[scale=0.35]{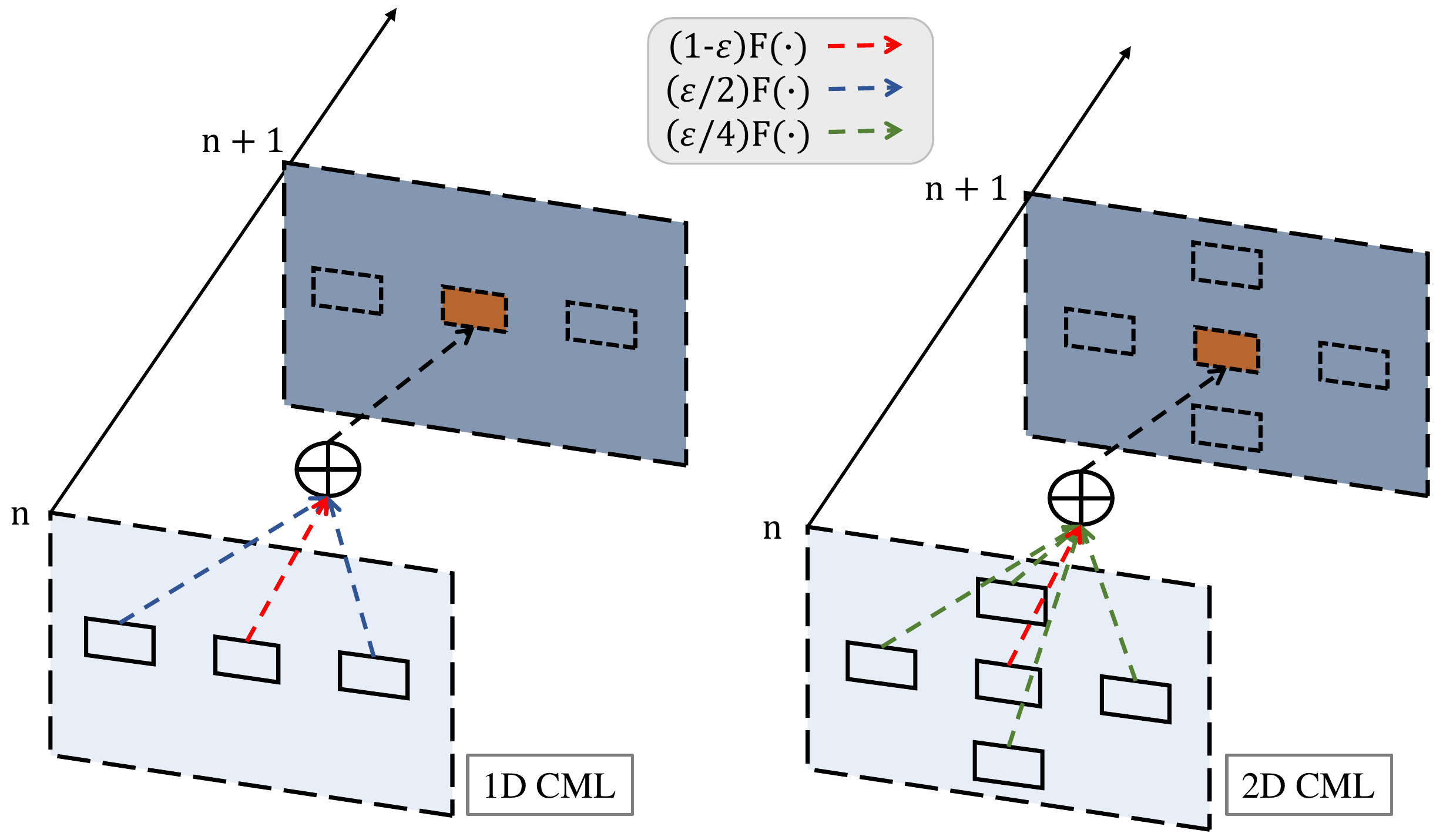}
	\caption{The diagram of the 1D and 2D CML models.}
	\label{Fig:1d2dCML}	
\end{figure}

\subsection{Two CML Models}
\label{SubSec:cml}
The first CML model, as depicted in Fig.~\ref{Fig:1d2dCML}, is proposed by Kaneko in \cite{kaneko1989pattern} and is one of the classic platform for studying spatiotemporal chaos. 
\begin{Definition} \cite{kaneko1992overview} The general one-dimensional nearest-neighboring CML model is described by
\begin{equation*}
x_{n+1}^u=(1-\varepsilon)\textnormal{F}(x_{n}^u)+\frac{\varepsilon}{2}\left[\textnormal{F}(x_{n}^{u+1})+\textnormal{F}(x_{n}^{u-1})\right], 
\end{equation*}
where $n=1,2,...,$ is the time index, $u=1, 2, ..., R,$ is the space index, $x_{n}^u$ is the state value of the  $u$-th node at time $n$, $\varepsilon \in (0,1)$ is the coupling parameter, and $F$ is a chaotic map.  The periodic boundary condition is $x_{n}^{0}=x_{n}^{R}$.
\end{Definition}

In the lattice, the coupling parameter $\varepsilon$ between the nodes plays an important role in maintaining the complex dynamic behavior of the model. A tiny change in one node can affect other nodes and even be diffused to the whole lattice after some iterations. To further enhance the dyamic complexity, the one-dimensional CML model is extended to the two-dimensional version as follows.
\begin{Definition} The two-dimensional CML model is defined by
\begin{IEEEeqnarray}{rCl}
\label{Eq:2dCML}
x_{n+1}^{u,v}  &=& (1-\varepsilon)\textnormal{F}(x_{n}^{u,v})+\frac{\varepsilon}{4} \left[\textnormal{F}(x_{n}^{u+1,v})+\textnormal{F}(x_{n}^{u-1,v})  \right. \nonumber \\
&+ & \left. \textnormal{F}(x_{n}^{u,v+1})+\textnormal{F}(x_{n}^{u,v-1})\right],
\end{IEEEeqnarray}
where \modified{$u=1,2,...,R$} and \modified{$v=1,2,...,L$} are the row and column indexes of the nodes, respectively. The periodic boundary conditions are $x_{n}^{u+R,v}=x_{n}^{u,v}$ and  $x_{n}^{u,v+L}=x_{n}^{u,v}$. 
\end{Definition}

\subsection{Lyapunov Exponent}
\label{subsec:LE}
Regarding nonlinear dynamic systems, the Lyapunov exponent (LE) is the key to measure chaotic behaviors. The maximum LE of the 
system $x_{n+1}=\textnormal{F}(x_{n})$ is defined as \cite{wolf1985determining}
\begin{equation}
\begin{split}
\textnormal{LE}&={\lim_{n \to \infty}\frac{1}{n}}\ln\left| 
\prod_{m=0}^{n}\textnormal{F}'(x_m)\right|. 
\end{split}
\label{Eq:LEdef}
\end{equation}
\Lcomment{The value of LE is affected by the parameters in $\textnormal{F}$ and a positive} $\textnormal{LE}$ indicates that the dynamic system is chaotic if its orbits are globally bounded. Moreover, the larger the LE is, the more complicated chaotic behavior the system has. For higher-dimensional systems or coupled systems like CML, it can have multiple LEs, and a strictly positive maximum LE indicates chaos.

For cryptographic confusion and diffusion suggested by Shannon \cite{shannon1949communication}, it is desirable to run a chaotic system with a  large LE, so that any tiny change of the system parameters will spread out and be amplified gradually. For the requirement of pseudo-randomness \cite{nian2011pseudo}, it is necessary to run the underlying chaotic system with a large LE so as to minimize the correlation between orbits produced by proximal parameters.

\subsection{Non-uniform Randomness and Non-uniform Pseudo-Randomness}
\label{subsec:prd}
Let $x_{n+1} = \textnormal{F}(x_n)$ be an iterative chaotic system with value range $(0,1)$. 
\Lcomment{Given the initial condition and system parameters, the orbits of this system are said to follow a pseudo-random distribution. Conceptually, by iterating $ \textnormal{F}$, the distribution function of $\{x_i\}_{i=0}^{\infty}$, $\textnormal{D}(\textnormal{F}, x_0)$, can be obtained.} 
\begin{Definition} 
\label{def:prd}
$\textnormal{D}(\textnormal{F}, x_0)$ is a pseudo-random distribution ({PRD}) if there exists a non-uniform true random distribution with the same density and there is no probabilistic polynomial-time (p.p.t.) algorithm that can distinguish them with a non-negligible probability.
\end{Definition}

In the studies of true random number generation, it is common to employ certain entropy extractor to squeeze uniform true random numbers from non-uniform true random source \cite{von1951random}. While in the studies of pseudo-random number generation, there is no explicit tool that allows one to extract uniform pseudo-random bits (the functionality of PRNG) from a non-uniform pseudo-random source with an  unknown distribution, and the introduction of PRD could be seen as one particular way to bridge this gap, as will be further discussed in Sec.~\ref{subsec:randomness}. 
In fact, it is easy to build a PRD by inversely sampling the output of a PRNG \cite{devroye1986sample}, \Lcomment{but not vice versa. From this  point of view,} it is reasonable to believe that, like true random numbers obtained by squeezing non-uniform true random source, uniform pseudo-random numbers may be obtained by manipulating and compressing certain non-uniform PRD.

\subsection{Independence Test}
\label{SubSec:test}
\begin{Definition} 
Let $x$ and $y$ be two random variables \Lcomment{that obey two random distributions}, and $x_{i}$ and $y_{i}$ ($i \in [1, K]$) be $K$ independent observations of $x$ and $y$, respectively. The Pearson correlation coefficient $k_{xy}$ between $x$ and $y$ is 
\begin{equation*}
k_{xy}=\frac{ \sum\limits_{i=1}^Kx_iy_i-K\bar{x}\bar{y}}
{\sqrt{\left( \sum\limits_{i=1}^Kx_i^2-K\bar{x}^2\right)\left( \sum\limits_{i=1}^Ky_i^2-K\bar{y}^2\right)}},
\end{equation*}
where $\bar{x}=\dfrac{1}{K}\sum\limits_{i=1}^Kx_i$ and $\bar{y}=\dfrac{1}{K}\sum\limits_{i=1}^Ky_i$.
\end{Definition}
If $x$ and $y$ are independent of each other, then Fisher's transformation of $k_{xy}$,
\begin{equation*}
\label{Eq:test}
D=\dfrac{\sqrt{K-3}}{2}\ln\left|\dfrac{1+k_{xy}}{1-k_{xy}}\right|,
\end{equation*}
will approximately follow the standard normal distribution. By setting a significance level $\alpha$, one can compare whether the empirical value D falls within the confidence interval $(\Phi^{-1}(\alpha/2), -\Phi^{-1}(\alpha/2))$, where $\Phi$ is the cumulative distribution function of the standard normal distribution. For a more reliable result, multiple tests should be applied and other tests, such as the chi-square test of independence and Kolmogorov-Smirnov test, should also be used. More importantly, one can apply this kind of tests between a PRD and a true random distribution, provided that no p.p.t. algorithm can distinguish this PRD from a certain true random distribution. \modified{Also, as will be used later in Sec.~\ref{subsec:randomness}, the test can be applied to two PRDs if both of them cannot be distinguished from true random distributions.}

\section{From Chaos to Pseudo-Randomness}
\label{sec:method}

With the preliminary knowledge introduced above, a theoretic analysis of the $2$D CML is presented in this section. Firstly, the analytical expression of all the LEs of the $2$D CML is deduced in Sec.~\ref{subsec:LEAna}, which will be used to enforce the system to run in fully developed chaotic state by using appropriate parameters. Secondly, it will discuss how to extract pseudo uniform bits with minimum bias from a pseudo-random distribution. Based on this result, an end-to-end extraction algorithm $\mathbf{E}$ will be designed to distill random bits from the orbits of $2$D CML with theoretical support.

\subsection{Lyapunov Exponent Formula of the $2$D CML}
\label{subsec:LEAna}

As stated in Sec.~\ref{Sec:intro}, for either designing chaotic ciphers or producing random numbers, one should carefully set the parameters to make the chaotic system to operate in a fully developed chaotic mode. 
\Lcomment{As stated in Sec.~\ref{subsec:LE},  this can be accomplished by carefully selecting the parameter(s) of the chaotic system.}
Here,  since LE is defined asymptotically as a limit, the synchronization of node values will be used  to derive an analytical expression of all LEs for the $2$D CML.

To begin with, convert the $2$D CML model with $R$ rows and $L$ columns to a one-dimensional model by rearranging the nodes according to the order from left to right and from top to bottom. Thus, Eq.~(\ref{Eq:2dCML}) is converted to
\begin{equation}
\label{Eq:2dto1d}
\begin{split}
x_{n+1}^{(u-1)L+v}=&(1-\varepsilon)\textnormal{F}(x_{n}^{(u-1)L+v})+\frac{\varepsilon}{4}\left[\textnormal{F}(x_{n}^{u\times L+v})\right. \\
&+\textnormal{F}(x_{n}^{(u-2)L+v})+\textnormal{F}(x_{n}^{(u-1)L+v+1})\\
&\left.+\textnormal{F}(x_{n}^{(u-1)L+v-1})\right].
\end{split}
\end{equation}
Correspondingly, the periodic boundary conditions will be changed to  $x_{n}^{(u+R)L+v}=x_{n}^{u\times L+v}$ and $x_{n}^{(u-1)L+v+L}=x_{n}^{(u-1)L+v}$. 
All the node values in the converted model can be arranged  as an $(R\times L)$-dimensional column vector, $\textbf{z}_{n}=[x_{n}^{1},x_{n}^{2},\cdots,x_{n}^{L},x_{n}^{L+1},x_{n}^{L+2},\cdots,x_{n}^{2L}, \cdots,x_{n}^{R\times L}]^{T}$. 
Similarly to the method used for the $1$D CML \cite{ding1997stability}, one can differentiate Eq.~(\ref{Eq:2dto1d}) and evaluate the derivatives along their synchronized trajectories. Upon synchronization, all the entries of $\textbf{z}_{n}$ become equal, i.e., 
$ x_{n}^{1}=x_{n}^{2}=\cdots=x_{n}^{L}=x_{n}^{L+1}=x_{n}^{L+2}=\cdots=x_{n}^{2L}= \cdots=x_{n}^{R\times L}=x_{n}$. 

Along the synchronized trajectory, one has the derivatives of $\textnormal{F}$ as
\begin{equation}
\label{Eq:2dCMLtrajectory}
\begin{split}
&\textnormal{F}'(x_{n}^{(u-1)L+v})=\textnormal{F}'(x_{n}^{u\times L+v})=\textnormal{F}'(x_{n}^{(u-2)L+v})\\
&=\textnormal{F}'(x_{n}^{(u-1)L+v+1})=\textnormal{F}'(x_{n}^{(u-1)L+v-1})=\textnormal{F}'(x_{n}),
\end{split}
\end{equation}
and the differentials of the 2D CML are
\begin{equation}
\label{Eq:2dCMLdifferentiate}
\begin{split}
\delta(x_{n+1}^{(u-1)L+v})=&
(1-\varepsilon)\textnormal{F}'(x_{n}^{(u-1)L+v})\delta(x_{n}^{(u-1)L+v})\\&+\dfrac{\varepsilon}{4}\left[\textnormal{F}'(x_{n}^{u\times L+v})\delta(x_{n}^{u\times L+v})\right.\\
&+\textnormal{F}'(x_{n}^{(u-2)L+v})\delta(x_{n}^{(u-2)L+v})\\
&+\textnormal{F}'(x_{n}^{(u-1)L+v+1})\delta(x_{n}^{(u-1)L+v+1})\\
&\left.+\textnormal{F}'(x_{n}^{(u-1)L+v-1})\delta(x_{n}^{(u-1)L+v-1})\right].
\end{split}
\end{equation}
By incorporating Eq.~(\ref{Eq:2dCMLtrajectory}), Eq.~(\ref{Eq:2dCMLdifferentiate}) can be written as
\begin{IEEEeqnarray}{rCL}
\delta(x_{n+1}^{(u-1)L+v})&=&\textnormal{F}'(x_{n}) \left[(1-\varepsilon)\delta(x_{n}^{(u-1)L+v}) \right. \nonumber \\ 
&& +\dfrac{\varepsilon}{4}(\delta(x_{n}^{u\times L+v})+\delta(x_{n}^{(u-2)L+v})  \nonumber \\
&& \left. +\delta(x_{n}^{(u-1)L+v+1})+ \delta(x_{n}^{(u-1)L+v-1}))\right]. \IEEEeqnarraynumspace
\label{Eq:2dCMLsimply}
\end{IEEEeqnarray}

\Lcomment{Applying Eq.~(\ref{Eq:2dCMLsimply}) to all the ($R\times L$) elements produced at time instants ($n+1$) and $n$, i.e., $\textbf{z}_{n+1}$ and $\textbf{z}_{n}$}, one can get a matrix form of Eq.~(\ref{Eq:2dCMLdifferentiate}), as,
\begin{equation*}
\delta\textbf{z}_{n+1}=\textbf{J}_{n}\delta\textbf{z}_{n},
\end{equation*}
where $\textbf{J}_{n}$ is the Jacobin matrix satisfying $\textbf{J}_{n}=\textnormal{F}'(x_{n})\textbf{K}$ and $\textbf{K}$ is an $ (R\times L)\times (R\times L)$ circulant matrix in the form of
\begin{equation*}
 \mathbf{K} = 
\begin{bmatrix}
\mathbf{A_{1}} & \mathbf{A_{2}}  &\cdots & \mathbf{A_{R}} \\
\mathbf{A_{R}} & \mathbf{A_{1}}  &\cdots & \mathbf{A_{R-1}} \\
\vdots         & \vdots          &\ddots & \vdots\\
\mathbf{A_{2}} & \mathbf{A_{3}}  &\cdots & \mathbf{A_{1}}  \\
\end{bmatrix}_{(R\times L)\times (R\times L)}.
\end{equation*}
Here, the matrices $\mathbf{A_1}, \mathbf{A_2}, \cdots, \mathbf{A_R}$ are given by
\begin{equation*}
\mathbf{\mathbf{A}_{1}}=
\begin{bmatrix}
1-\varepsilon &\varepsilon/4 &0 &\cdots &\varepsilon/4 \\
\varepsilon/4 &1-\varepsilon &\varepsilon/4 &\ddots & 0\\
0 &\varepsilon/4 &1-\varepsilon &\ddots & 0\\
\vdots &\vdots  &\ddots &\ddots &\varepsilon/4\\
\varepsilon/4 &  0 &\cdots &\varepsilon/4 &1-\varepsilon  \\
\end{bmatrix}_{L\times L},
\end{equation*}
\begin{equation*}
\mathbf{\mathbf{A}_{2}=\mathbf{A}_{R}}=
\begin{bmatrix}
\varepsilon/4 & \quad  &\quad \\
\quad  & \ddots &\quad  \\
\quad & \quad &\varepsilon/4 \\
\end{bmatrix}_{L\times L},
\end{equation*}
and
\begin{equation*}
\mathbf{\mathbf{A}_{3}=\mathbf{A}_{4}=\cdots =\mathbf{A}_{R-1}}=
\begin{bmatrix}
0       & \cdots    & 0 \\
\vdots  &\ddots     & \vdots\\
0       & \cdots    & 0 \\
\end{bmatrix}_{L\times L}.
\end{equation*}

\Lcomment{To calculate the LEs of the 2D CML, one needs to multiply all the Jacobin matrices}. 
Let $\lambda$ represent an eigenvalue of the matrix $\textbf{K}$ and denote $\textbf{G}=\textbf{J}_{1}\times \textbf{J}_{2}\times \cdots \times \textbf{J}_{n}=\textbf{K}^{n} \cdot \left(\prod_{m=1}^{n}\textnormal{F}'(x_{m})\right)$. It is easy to verify that the eigenvalue of $\textbf{G}$ is $\lambda^{n} \cdot \left(\prod_{m=1}^{n}\textnormal{F}^{'}(x_{m})\right)$.
According to  Eq.~(\ref{Eq:LEdef}), the LEs of $2$D CML are given by the following formula, parametrized by $\lambda$:
\begin{IEEEeqnarray}{rCL}
\label{Eq:LE}
\textnormal{LEs}&=&
{\lim_{n \to \infty}\frac{1}{n} }\ln \left| \textbf{J}_{1}\times \textbf{J}_{2}\times \cdots \times \textbf{J}_{n} \right| \nonumber \\
&=& {\lim_{n \to \infty}\frac{1}{n} }\ln \left| \lambda^{n}\prod_{m=1}^{n}\textnormal{F}'(x_m) \right|  \nonumber  \\
&=& \lim_{n \to \infty}\frac{1}{n} \ln\left| \prod_{m=1}^{n}F'(x_{m})\right|
+\ln\left| \lambda \right|.
\end{IEEEeqnarray}

\Lcomment{Note that the first term in Eq.~(\ref{Eq:LE}) is precisely the LE of the local chaotic map $\textnormal{F}$, that is, $ \lim_{n \to \infty}\frac{1}{n} \ln\left| \prod_{m=1}^{n}F'(x_{m})\right| = \textnormal{LE}_{\textnormal{F}}$.}
{
So, to calculate the LE of the $2$D CML, one has to compute the second term $\ln\left| \lambda \right|$.}
{Here, $\lambda$ represents any eigenvalue of $\textbf{K}$ and $\textbf{K}$ is a block circulant matrix with most of the blocks being zero (i.e., $\textbf{A}_{3}=\cdots=\textbf{A}_{R-1}=\textbf{0}$). Its characteristic polynomial is \cite{hongbo2013characteristic}}
\begin{equation}
\label{Eq:polynomial}
\prod_{r=0}^{R-1}\left|\textbf{A}_{1}+\textbf{A}_{2}\omega_{r}+\textbf{A}_{R}\omega_{r}^{R-1}- \lambda\textbf{I}\right|, 
\end{equation}
where $\textbf{I}$ is the identity matrix and 
\begin{equation*}
    \omega_{r}=\exp\left(i  \dfrac{2\pi r}{R}\right)=\cos\dfrac{2\pi r}{R}+i  \sin\dfrac{2\pi r}{R},
\end{equation*}
for \Lcomment{$r=0,1,\cdots,R-1$}. Expanding the innner part of Eq.~(\ref{Eq:polynomial}), one gets
\begin{IEEEeqnarray*}{rCl}
\textbf{A} 
&=& \textbf{A}_{1}+\textbf{A}_{2}\omega_{r}+\textbf{A}_{R}\omega_{r}^{R-1}- \lambda\textbf{I} \\
&=& \begin{bmatrix}
j& \varepsilon/4 & 0 &\cdots &\varepsilon/4 \\
\varepsilon/4& j &  \varepsilon/4 &\cdots & 0 \\
\vdots & \vdots & \vdots &\ddots & \vdots\\
\varepsilon/4 & 0 & 0 & \cdots &j\\
\end{bmatrix}_{L\times L},
\end{IEEEeqnarray*}
where $j=1-\varepsilon+\dfrac{\varepsilon}{4}\omega_{k}+\dfrac{\varepsilon}{4}\omega_{k}^{R-1}-\lambda$. It is clear that  $\textbf{A}$ is still a circulant  matrix. Applying the same technique as used above \cite{hongbo2013characteristic}, one obtains 
\begin{equation}
\label{Eq:eigenvalues}
 \lambda = 1-\varepsilon+\dfrac{\varepsilon}{4}\omega_r+\dfrac{\varepsilon}{4}\omega_r^{R-1}+\dfrac{\varepsilon}{4}\mu _{l}+\dfrac{\varepsilon}{4}\mu _{l}^{L-1},
\end{equation}
where $\mu_{l}=\exp\left( i\dfrac{2\pi l}{L}\right)=\cos\dfrac{2\pi l}{L}+i\sin\dfrac{2\pi l}{L},$ with $l = 0, 1, \cdots, L-1$.

{\modified{Substitute Eq.~(\ref{Eq:eigenvalues}) into Eq.~(\ref{Eq:LE}) gives all the $(R\times L)$ LEs of $2$D CML, i.e.,}}
\begin{equation} 
\label{Eq:caculatele}
  \textnormal{LEs}=\textnormal{LE}_{\textnormal{F}}+\ln\left|1-\varepsilon+\dfrac{\varepsilon}{2}\left(\cos\frac{2\pi r}{R}+\cos\frac{2\pi l}{L}\right)\right|,
\end{equation}
for $r = 0,1,\cdots, R-1$ and $l=0,1,\cdots,L-1$.
From this analytical formula, it is easy to get the  following results.

\begin{Theorem}
\label{MLEthero}
The maximum Lyapunov exponent of the $2$D CML model is determined \modified{by the local chaotic map $\textnormal{F}$}.
\end{Theorem}
\begin{proof}
According to Eq.~(\ref{Eq:caculatele}), the maximum LE of $2$D CML is given by  $\textnormal{LE}=\textnormal{LE}_{\textnormal{F}}$, when $r=0$ and $l=0$.
\end{proof}

\begin{Corollary}
The maximum Lyapunov exponent is independent of the size of the $2$D CML model.
\end{Corollary}
\begin{Corollary}
\label{increasingMLE}
In $2$D CML, increasing $\textnormal{LE}_{\textnormal{F}}$ will increase the maximum Lyapunov exponent  of the whole 2D lattice.
\end{Corollary}

\subsection{Pseudo-Randomness Extraction From the $2$D CML}
\label{subsec:randomness}
According to the above analyses, with appropriate selection of the local chaotic map $\textnormal{F}$, the maximum Lyapunov exponent of the $2$D CML will be positive and the orbits of the whole system will only run in fully developed chaotic state.

Given the initial conditionals of the $2$D CML and the system parameters of $\textnormal{F}$, if any, the distribution of the orbits is governed by a PRD\footnote{This kind of chaotic system is referred to as a $2$D CML instance.}. 
Without loss of generality, assume that this PRD is not uniform. This is because only few local chaotic maps are known to have a uniform density, for example the tent map \cite{heidel1990existence}. And, even if a uniform local map is used, the overall density is likely to be uneven after coupling the local maps together in the structure of $2$D CML.
Similarly to sampling a random variable from a true random distribution, the way of taking a sample from this PRD is to iterate the $2$D CML. 

The following discussion shows how to squeeze uniform bits from the non-uniform PRD that governs the $2$D CML's orbits. This discussion starts with a general result that holds for both pseudo-random distribution and true random distribution.
\begin{Theorem}
\label{theorem:uniformtheorem}
For a random (or pseudo-random) distribution in $[0,1]$, assume that the density function has bounded first-order derivative.
For any sample $x=0.w_1w_2\cdots w_{z}$ ($w_i \in \{0,1\}$ and $i \in [1, z]$) from this distribution, one has
\begin{equation*}
    \lim_{z \to \infty} 
    P(w_{z}=0) = \lim_{z \to \infty}  P(w_{z}=1).
\end{equation*}
\end{Theorem}
\begin{proof}
See the appendix.
\end{proof}

\begin{figure}[ht]
	\centering
	\includegraphics[scale=0.45]{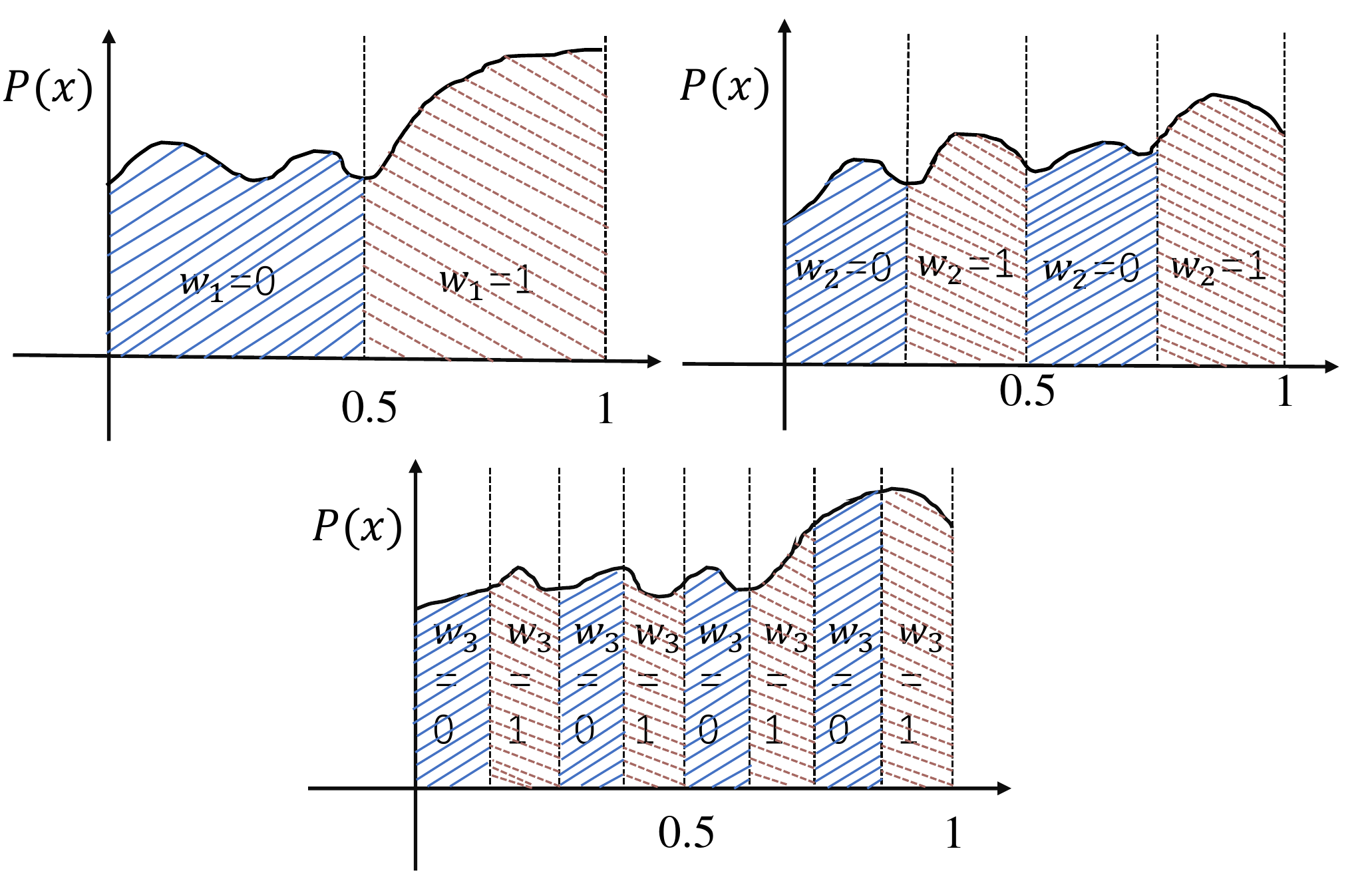}
	\caption{Illustration of Theorem~\ref{theorem:uniformtheorem}.}
	\label{Fig:uniformBit}
\end{figure}
\noindent \textbf{Remark}: This theorem can be intuitively understood from the examples depicted in Fig.~\ref{Fig:uniformBit}. For any (pseudo) random sample $x \in [0,1]$, if there is only $1$-bit ($z=1$) used to represent this number, then the difference of the probabilities of $P(w_1=0)$ and $P(w_1=1)$ equals the difference of the areas associated with the two slices in blue and red (as shown in the top-left of Fig.~\ref{Fig:uniformBit}). As $z$ goes larger (for example, $z=2$ and $z=3$ in Fig.~\ref{Fig:uniformBit}), the slices become smaller and their associated areas become similar. Thus, from the appendix, it is easy to see that $|P(w_z=0)- P(w_z=1)| = O(2^{-z})$.

\noindent \textbf{Remark}: Although the theorem is proved by assuming that the density function has a bounded first-order derivative, which implies continuity, it can be easily extended to a density with countably many points of discontinuity, which is always the case of finite precision representations of real chaotic orbits. 

Theorem~\ref{theorem:uniformtheorem} states the fact that, under fixed-point arithmetic, the farther away from the radix point, the more uniform the bits tend to be. 
For a given PRD determined by a $2$D CML instance, a naive way of using Theorem~\ref{theorem:uniformtheorem} to extract (pseudo) random bits is to take the rightmost bit of a chaotic orbit under finite precision representation. 
If the precision is $z$ bit, a single uniform bit with bias bounded by $O(2^{-z})$ can be obtained from one orbit. 
However, as discussed in Sec.~\ref{Sec:intro}, the efficiency of this method is too low to meet practical requirements. Nevertheless, as will be seen in the following, this drawback can be avoided by making use of more CML instances and Property~\ref{pro:uniformbias} given below. 

For two different $2$D CML instances, assume that their associated initial conditions and system parameters are selected independently. Consequently, the associated PRDs induced by these two instances are independent of each other, and so do the pseudo-random samples drawn from the two PRDs. For any sample $x=0.x_1x_2\cdots x_{z}$ drawn from the first PRD and $y=0.y_1 y_2\cdots y_{z}$ drawn from the second PRD, by letting $w = x_i + y_j \pmod{2}$, one can conclude that $w$ is more uniform than both $x_i$ and $y_j$ ($i, j \in [1, z]$). Specifically, the bias of $w$ is bounded by $O(2^{-(i+j)})$, as shown below.

\begin{Property}
\label{pro:uniformbias}
Let $x_i$ and $y_j$ be independent binary random bits with bounded bias $O(2^{-i})$ and $O(2^{-j})$, respectively. Then, $w= x_i+y_j \pmod{2}$ satisfies
\begin{IEEEeqnarray}{rCL}
|P(w=0)- P(w=1)| = O(2^{-(i+j)}). \nonumber
\end{IEEEeqnarray}
\end{Property}
\begin{proof}
See the appendix.
\end{proof}

Extending this property further, the following corollary can be easily obtained. From this corollary, it is easy to see that, with more independent PRDs, one can obtain bits with arbitrarily small bias. 

\begin{Corollary}
\label{Corollary:BetterUniform}
Let $w_{1}, w_{2},\cdots,w_{c}$ be $c$ mutually independent binary random variables, and each with bounded bias $O(2^{-k_i})~(i\in [1, c])$, and set $w = \sum\nolimits_{i=1}^c w_{i} \pmod{2}$. Then, one has
\begin{IEEEeqnarray}{rCL}
|P(w=0)- P(w=1)| = O(2^{-\left(\sum\nolimits_{i=1}^c k_{i}\right)} ). \nonumber
\end{IEEEeqnarray}   
\end{Corollary}

Property~\ref{pro:uniformbias} holds only when the samples $x$ and $y$ are drawn from independent (pseudo) random distributions, which requires the two CML instances be selected independently. By making use of the result about LE in Sec.~\ref{subsec:LEAna}, this assumption can be further relaxed. If the initial conditions or the system parameters of the CML instances are correlated but they both have large positive LE, then the correlation between their orbits will be dispersed fast after a few iterations. This transition period can be determined by applying the independence test from  Sec.~\ref{SubSec:test} to the orbits. In this way, the choice of the two $2$D CML instances can be made arbitrary if both CMLs have large positive $\textnormal{LEs}$ and the orbits in transition are discarded.

Based on the discussion above, one can improve the efficiency of pseudo-random bits generation for $z$ times by using two $2$D CML instances, while keeping the bias be $O(2^{-z})$. The end-to-end extraction Algorithm \ref{Alg:extraction} summarizes the details. 
\begin{algorithm}
\label{Alg:extraction}
\caption{Extraction Algorithm $\mathbf{E}$}
\KwIn{Two sets of initial conditions and system parameters 2D CML.}
\KwOut{Pseudo-random bits.}
\DontPrintSemicolon
%
\SetKwFunction{FModAdd}{\textnormal{ModAdd}}
\SetKwFunction{FMain}{\textnormal{Main}}
%
\SetKwProg{Fn}{Function}{:}{}
\Fn{\FModAdd{$x$, $y$}}{
    $x = 0.x_1x_2 \cdots x_z$\;
    $y = 0.y_1y_2 \cdots y_z$\;
    $w = w_1w_2 \cdots w_z$\;
    $w_i = x_i + y_{z+1-i} \pmod 2$\;
    \KwRet $w$\;
}
\SetKwProg{Fn}{Function}{:}{}
\Fn{\FMain}{
    \noindent \textbf{Step 1} Run the two CML instances and collect their corresponding orbits $\{x^i\}_{i=0}^K$ and $\{y^i\}_{i=0}^K$, and update the states\; 
    \noindent \textbf{Step 2} \modified{Perform the independence test on the collected orbits}\;
    \eIf{\textnormal{test passed}}{
      \While{\textnormal{necessary}}
      {
          Run the two instances to get ($x^0$, $y^0$)\; 
        ModAdd($x^0$, $y^0$)\; 
        }
    }
    {
        go back to \textbf{Step~1}\;
    }
\KwRet 0\;
}
\end{algorithm}

\section{Experimental Analysis}
\label{sec:SimTest}
In this section, numeric tests are carried out on different $2$D CML systems to assess the applicability and performance of the theoretic results developed in Sec.~\ref{sec:method}. Some popular chaotic systems are used as the local map $\textnormal{F}$ of Eq.~(\ref{Eq:2dCML}):
\begin{itemize}
    \item The Logistic map
    \begin{IEEEeqnarray*}{rCl}
    x_{n+1}  = \mu x_n (1-x_n), 
    \end{IEEEeqnarray*}
    where $\mu \in (0,4]$ and $x_n \in (0,1)$.
    
    \item The Tent map 
     \begin{IEEEeqnarray*}{rCl}
        x_{n+1}  = \left\{ \, 
        \begin{IEEEeqnarraybox}[][c]{l?s}
            \IEEEstrut
            \mu x_n,         & if $x_n\in (0,\modified{0.5})$, \\
            \mu (1-x_n),    & if $x_n\in [0.5,1)$,
            \IEEEstrut
        \end{IEEEeqnarraybox}
        \right.
    \end{IEEEeqnarray*}
    where $\mu \in (0,2]$ is the system parameter.
    
    \item The piecewise Logistic map (PLM)
      \begin{IEEEeqnarray*}{rCl}
        x_{n+1}  = \left\{ \, 
        \begin{IEEEeqnarraybox}[][c]{l?s}
            \IEEEstrut
            & $\mu N^2(x_n-\frac{i-1}{N})(\frac{i}{N}-x_n)$ , \\
                  &~~~~~~ if $\frac{i-1}{N}<x_n<\modified{\frac{i}{N}}$, \\
            & $1- \mu N^2 (x_n-\frac{i-1}{N})(\frac{i+1}{N}-x_n)$,\\
                  &~~~~~~ if $\frac{i}{N}<x_n<\frac{i+1}{N}$,
            \IEEEstrut
        \end{IEEEeqnarraybox}
        \right.
    \end{IEEEeqnarray*}
    where $\mu \in (0,4]$ is the system parameter and $N$ is the number of segments. 
\end{itemize}
It is worth mentioning that, other than the basic ones listed here, chaotic maps with more complex behavior, such as the ones suggested in \cite{hua2017designing,hua2015dynamic}, can also be  employed as the local map of the $2$D CML system. For simplicity, refer the resultant chaotic systems to as L-CML, T-CML, and PLM-CML, respectively.

\subsection{Lyapunov Exponents and Bifurcation Diagrams}
In this section, it is verified that the theoretic LE formula ~(\ref{Eq:caculatele}) of the $2$D CML is correct. Then, the associated bifurcation diagram is used to demonstrate that the orbits of the CML instance indeed run in full chaotic state with appropriate parameters. 

The reliable Wolf's scheme \cite{wolf1985determining} is used to calculate the LEs of the $2$D CML, and the result is plotted in Fig.~\ref{Fig:plotLEs} with $\mu = 4, 2, 4$ for L-CML, T-CML and PLM-CML, respectively. Here, the other parameters in Eq.~(\ref{Eq:2dCML}) are set to $\varepsilon =0.1$, $R= L = 8$, and the number of segments of PLM $N = 64$. It is clear that the theoretic scheme given by Eq.~(\ref{Eq:caculatele}) perfectly aligns with the numeric method.
\begin{figure*}[ht]
	\centering
	\centering
	\subfigure[LEs of L-CML]
	{
		\includegraphics[scale=0.32]{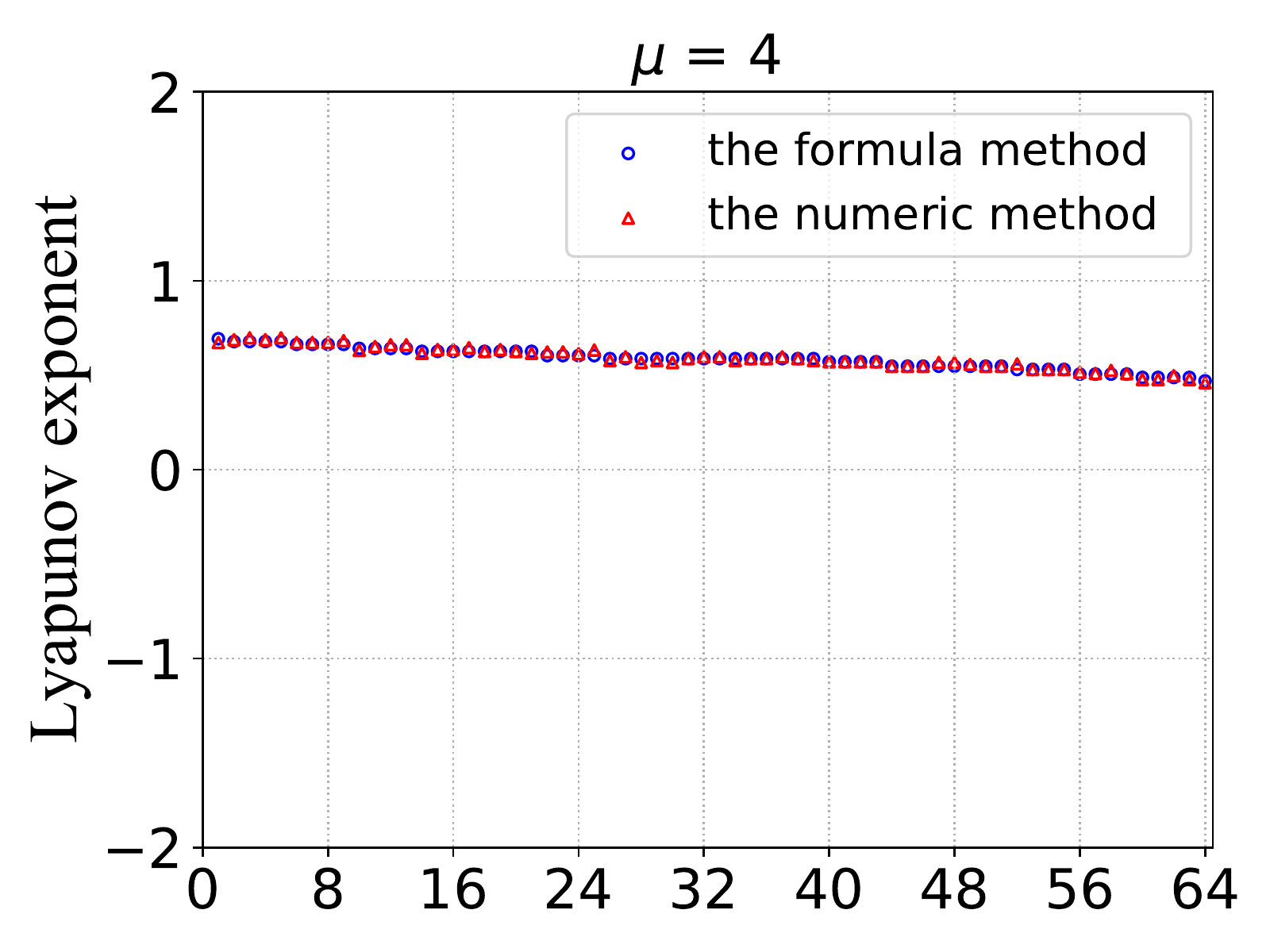}
	} 
	\subfigure[LEs of T-CML]
	{
		\includegraphics[scale=0.32]{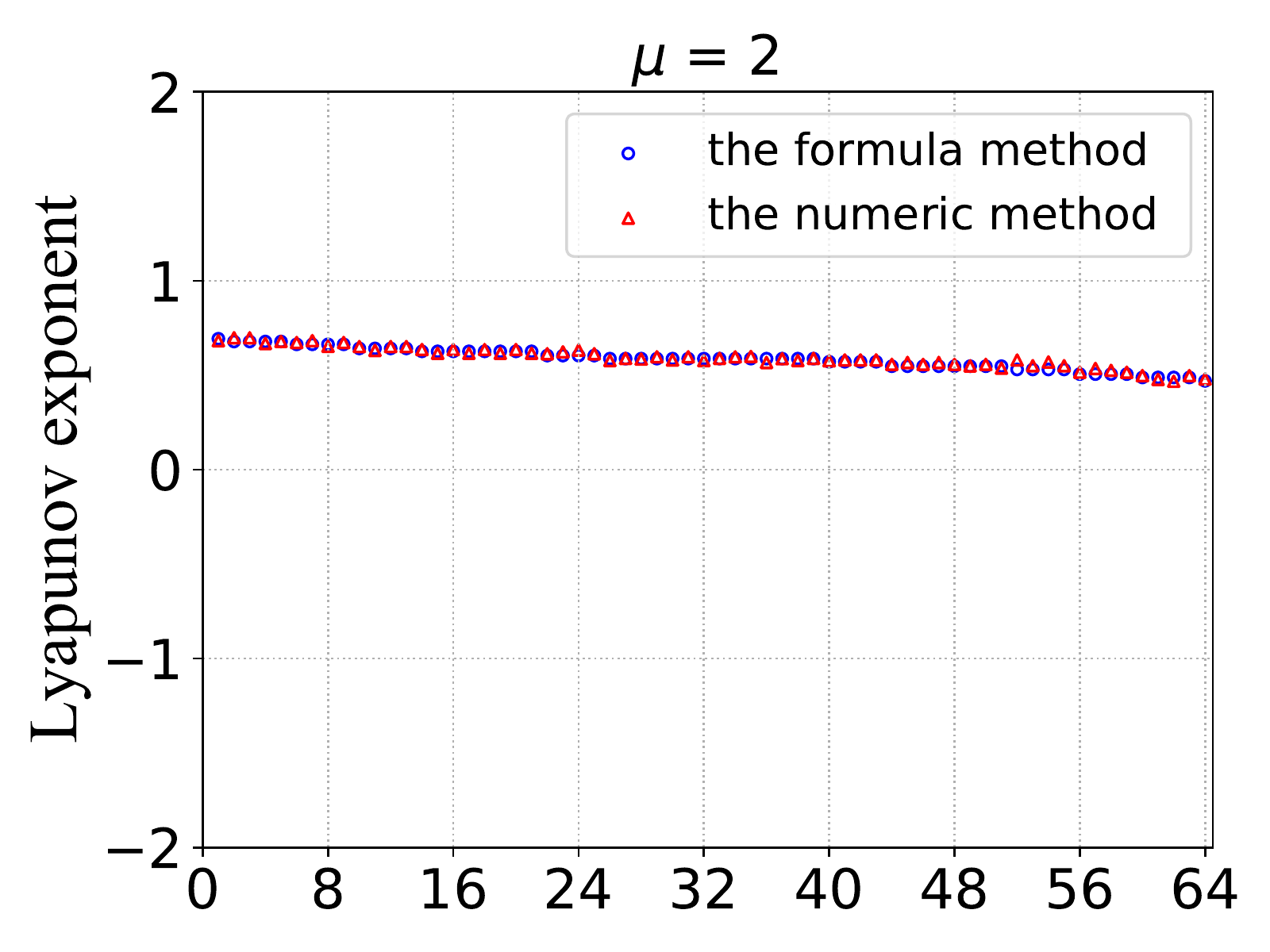}
	}
	\subfigure[LEs of PLM-CML]
	{
		\includegraphics[scale=0.32]{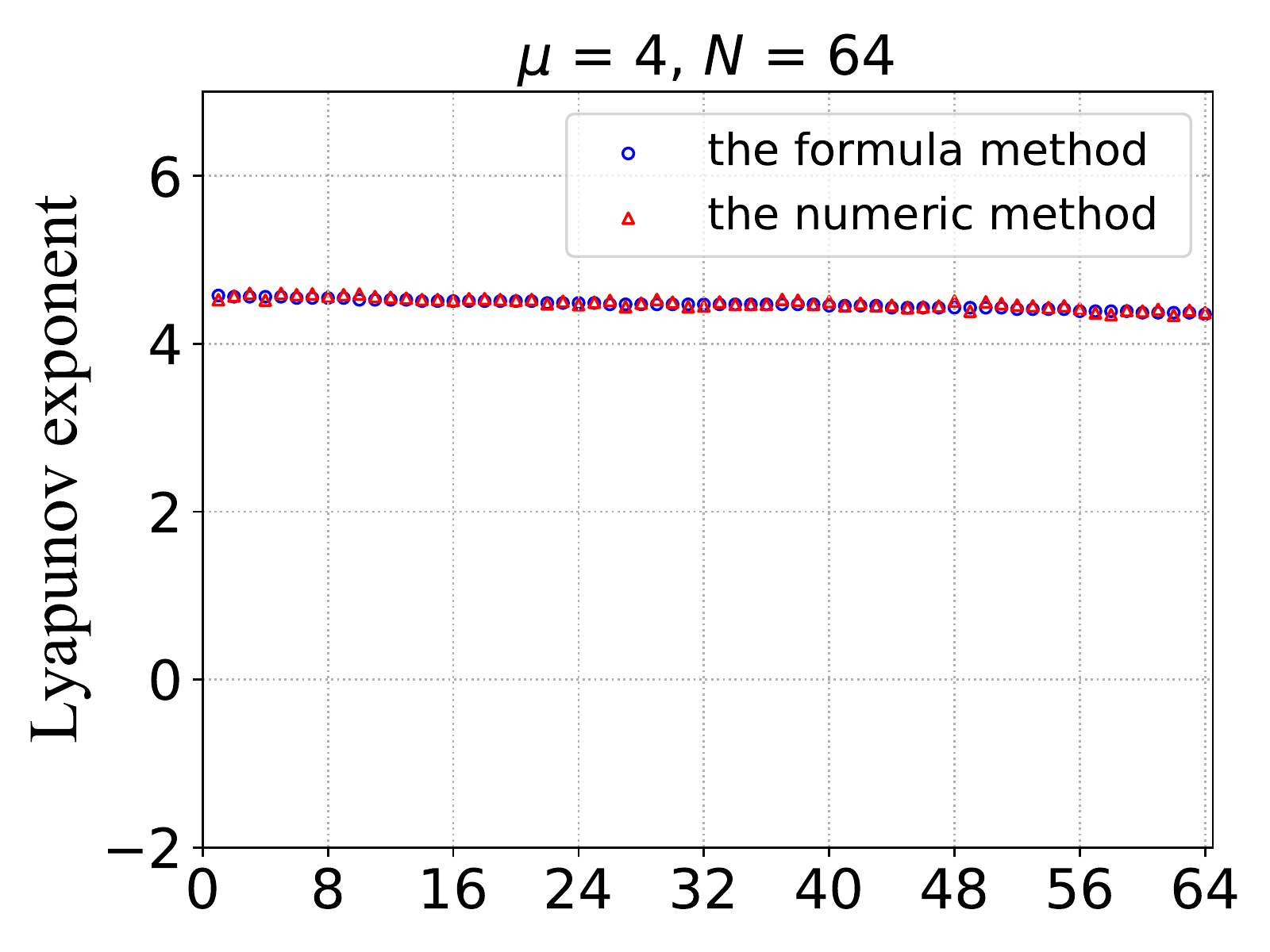}
	}
	\caption{LEs obtained from the numeric method in \cite{wolf1985determining} and theoretic result in Eq.~(\ref{Eq:2dCML}).}
	\label{Fig:plotLEs}
\end{figure*}

To demonstrate of chaos, the bifurcation diagram is used to plot output orbits of the system with respect to the change of the parameters. As shown in Theorem~\ref{MLEthero}, the maximum LE of a $2$D CML is solely determined by the local map $\textnormal{LE}_{\textnormal{F}}$. And, the orbits extracted from the first node (i.e., $r=l=0$ in Eq.~(\ref{Eq:2dCML})) achieves the maximum LE. 
By varying the parameter $\mu$, one can exact the orbits of the first node of L-CML, T-CML, and PLM-CML, respectively, and their bifurcation diagrams are shown in Fig.~\ref{Fig:1stBifurcation}. It is clear that the orbits of CML run in the chaotic state with an appropriate choice of the parameter of the local map. 
\begin{figure*}[ht]
	\centering
	\centering
	\subfigure[]
	{
		\includegraphics[scale=0.32]{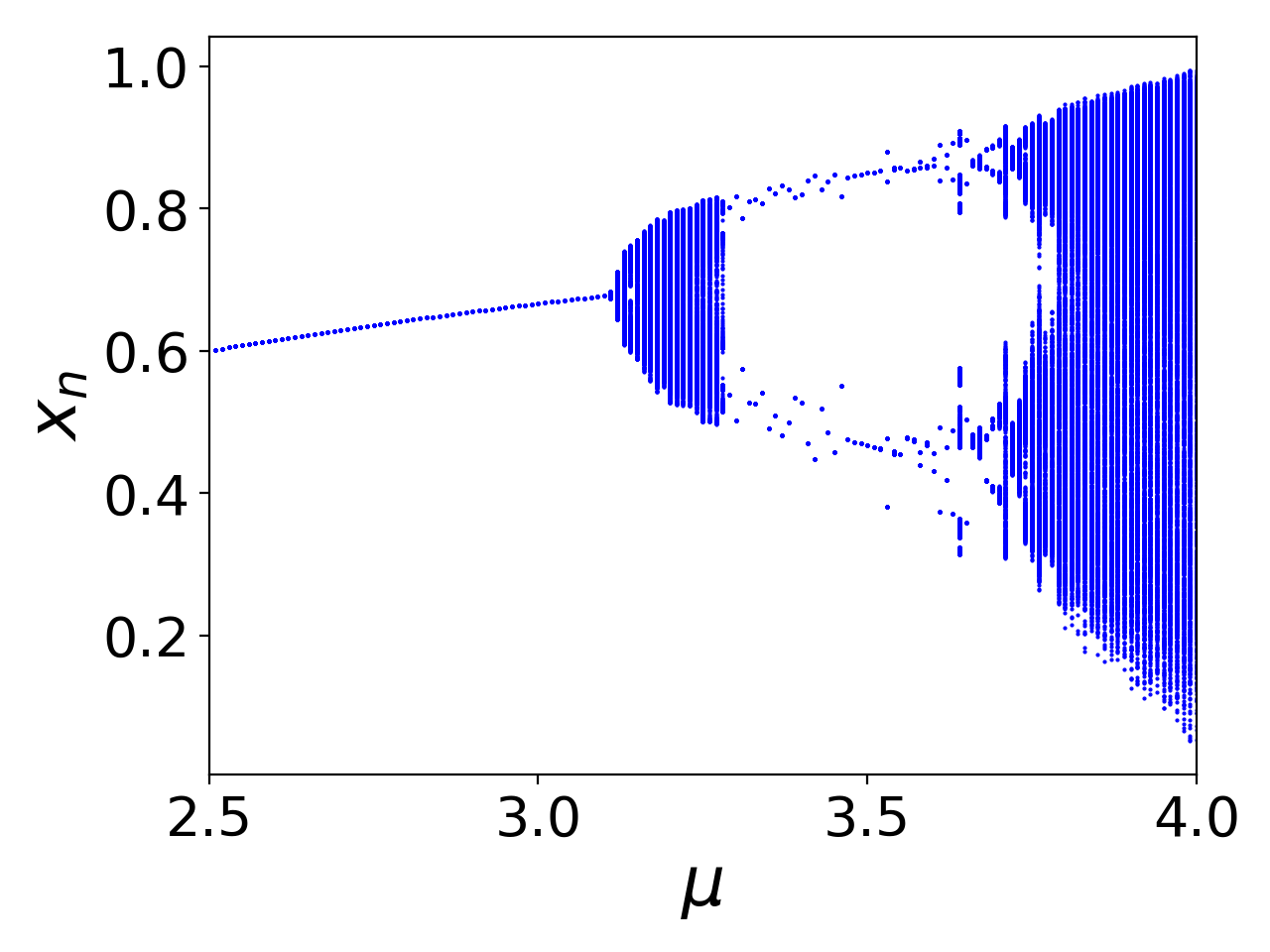}
	} 
	\subfigure[]
	{
		\includegraphics[scale=0.32]{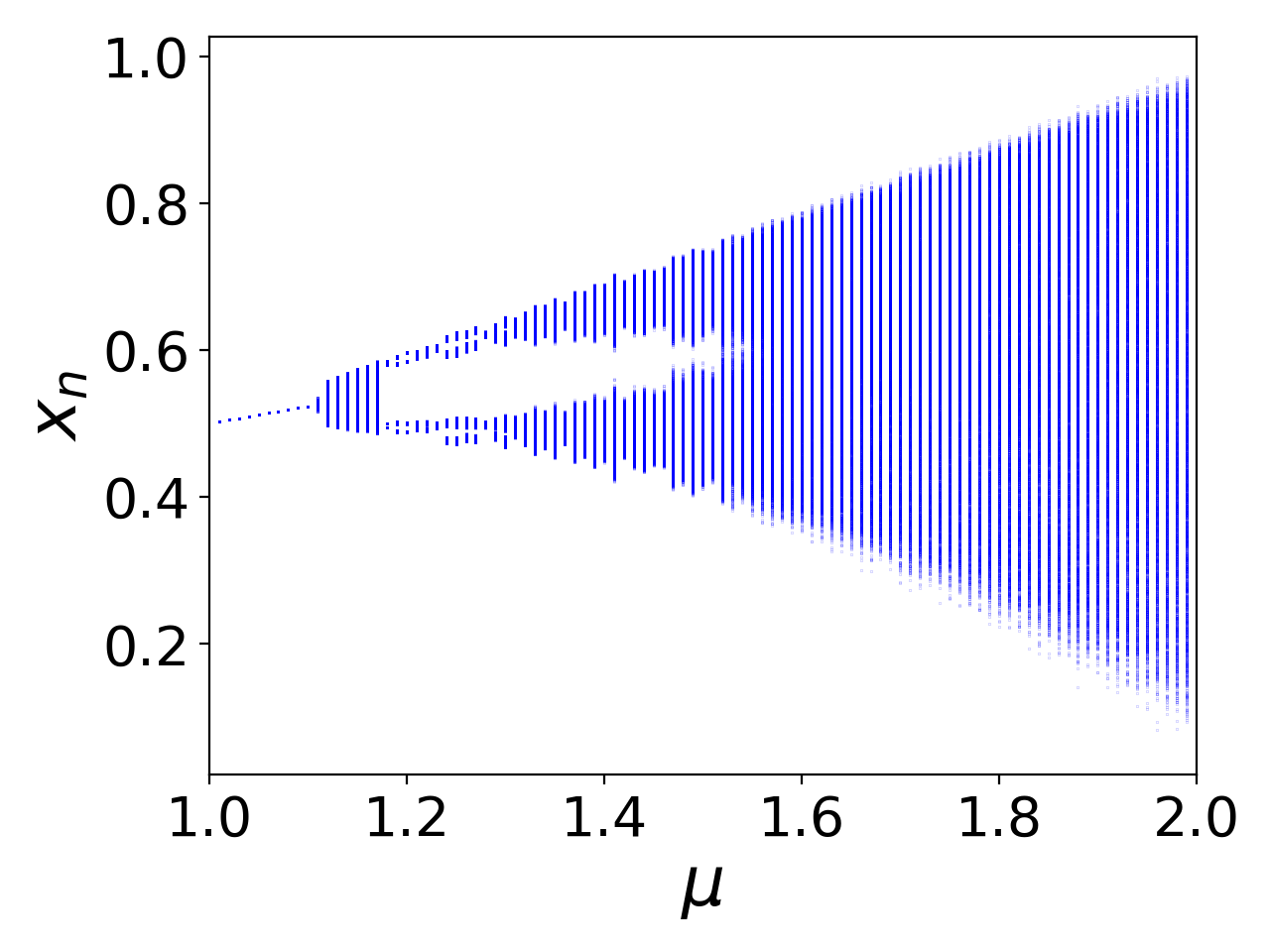}
	}
	\subfigure[]
	{
		\includegraphics[scale=0.32]{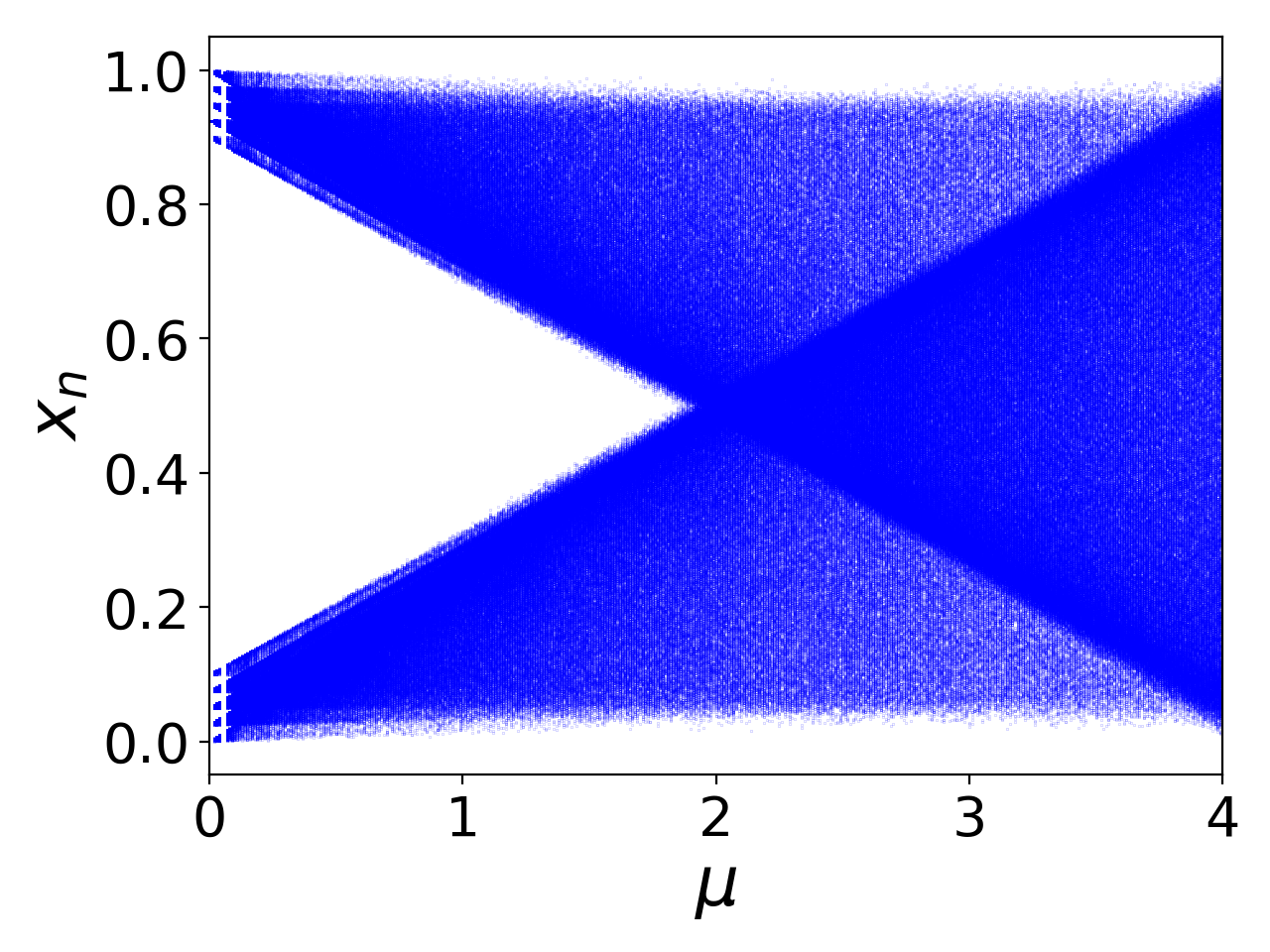}
	}
	\caption{Bifurcation diagram corresponding to the first node of L-CML, T-CML and PLM-CML, respectively.}
	\label{Fig:1stBifurcation}
\end{figure*}
Once again, it should be pointed out that, with more complex  chaotic systems, the $\textnormal{LE}_{\textnormal{F}}$ can be made larger \cite{hua2017designing,hua2015dynamic}, and the resulting chaotic performance will be better. This may bring up more benefits like a higher extraction rate of the pseudo-random numbers, but such possible improvement is left for future study. 
Moreover, other metrics, such as Kolmogorov entropy, fractal power spectra and correlation dimension, can be used to measure the performance of chaos, which likewise not the focus of this study.

\subsection{Randomness Tests}
To assess the results presented in Sec.~\ref{subsec:randomness}, the pseudo-random distributions introduced by L-CML, T-CML and PLM-CML are first plotted. Under the same parameter settings, the resulting distributions are shown in Fig.~\ref{Fig:distribution}. 
%
%
\modified{Note that the distribution of the Tent map is known to be uniform when $\mu =2$ \cite{heidel1990existence}, but the distribution of T-CML, like that of L-CML and PLM-CML, is not uniform. This finding supports the arguments in Sec.~\ref{Sec:intro}, suggesting the need for a theoretically sound approach for squeezing uniform random numbers from chaotic orbits.}
\begin{figure*}[ht]
	\centering
	\centering
	\subfigure[]
	{
		\includegraphics[scale=0.32]{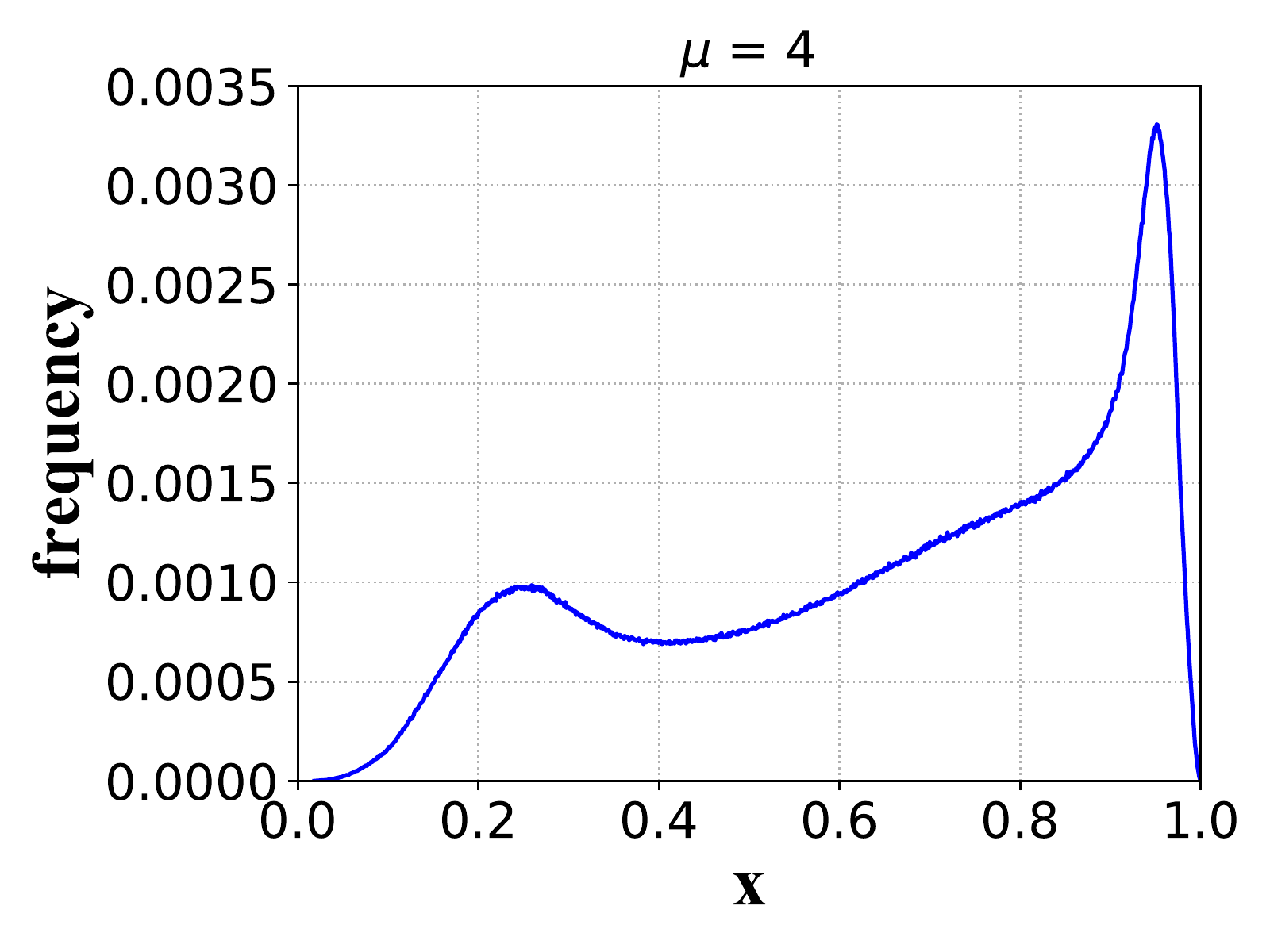}
	} 
	\subfigure[]
	{
		\includegraphics[scale=0.32]{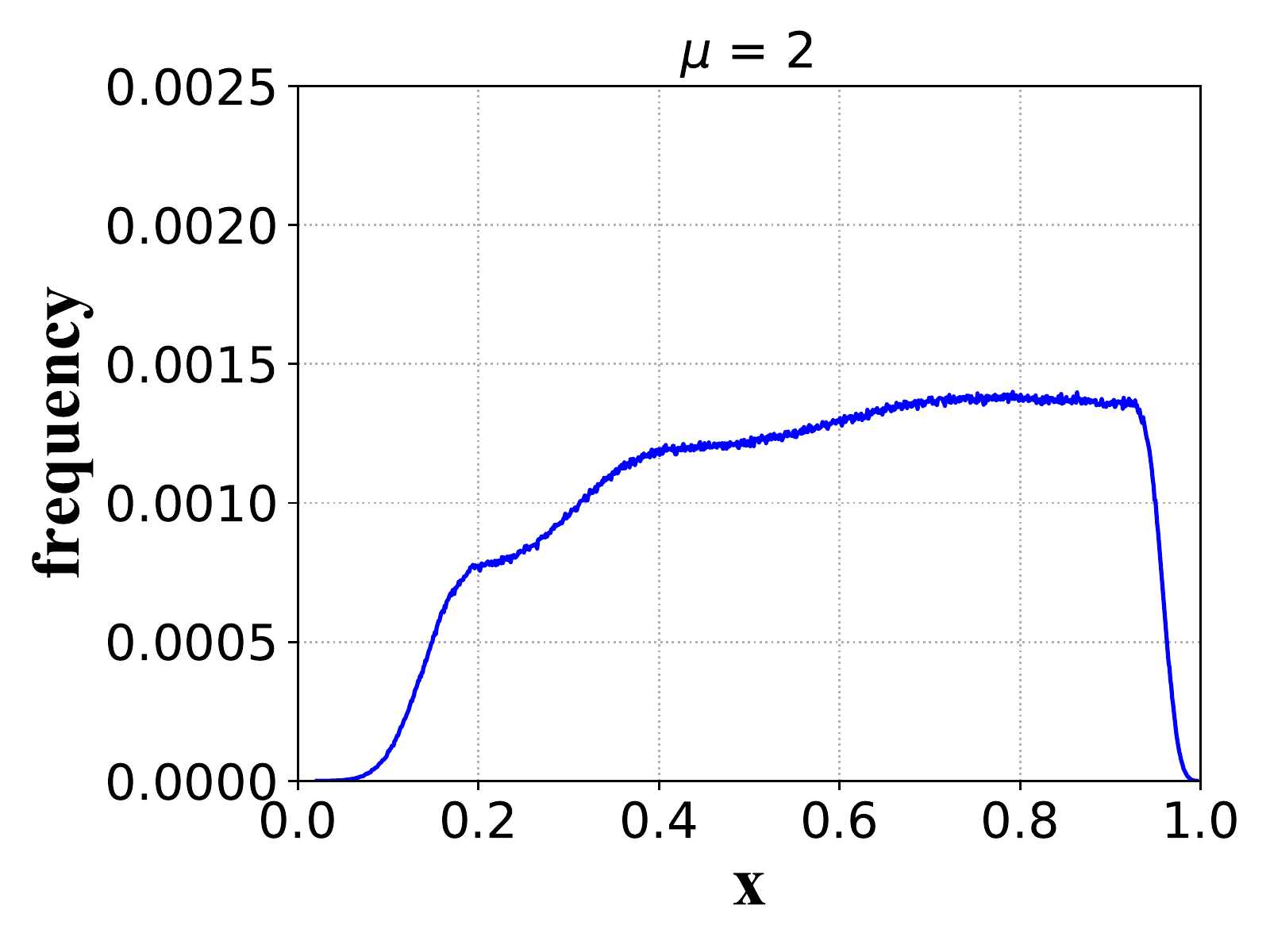}
	}
	\subfigure[]
	{
		\includegraphics[scale=0.32]{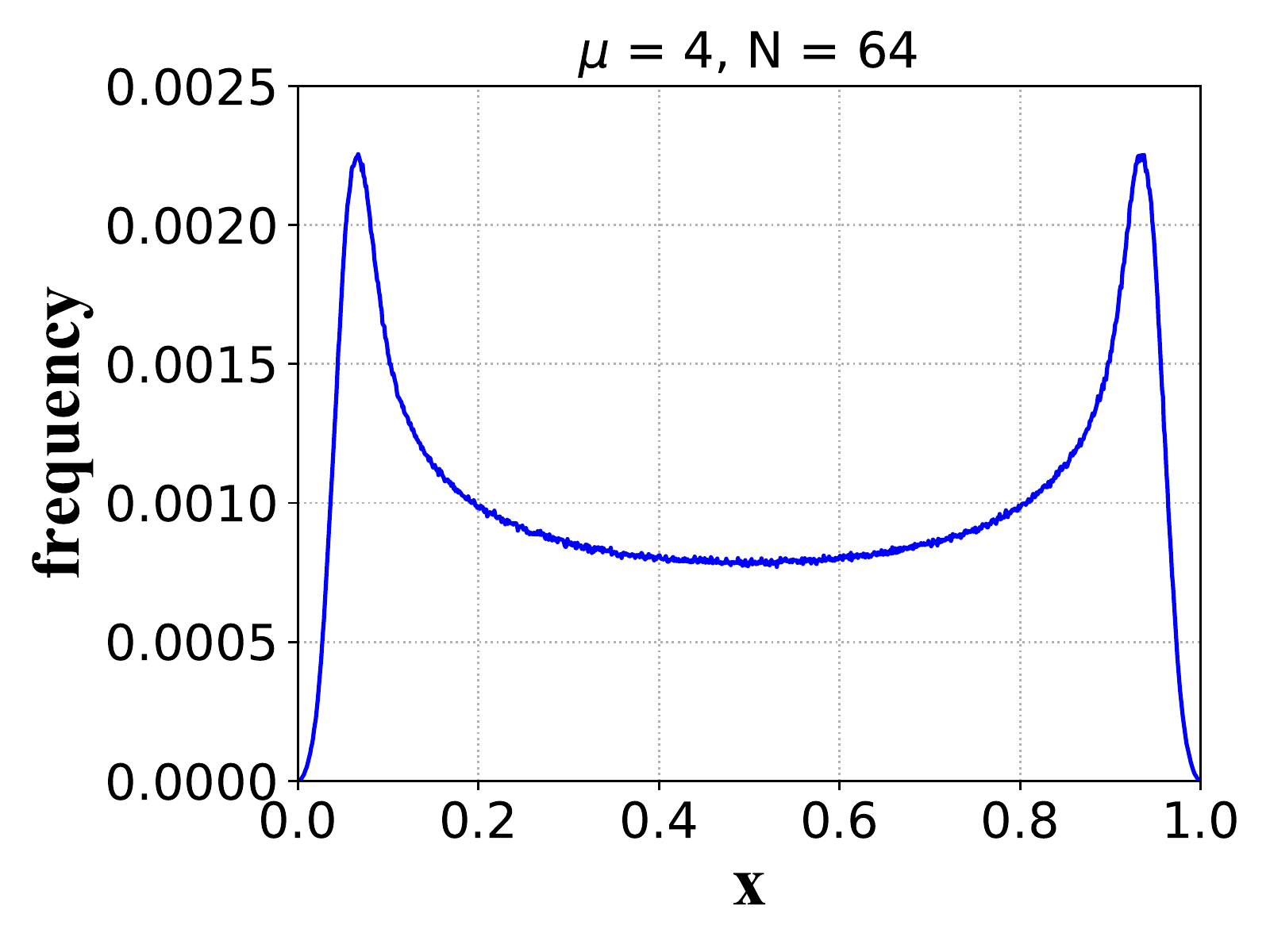}
	}
	\caption{Orbits distribution of L-CML, T-CML and PLM-CML, respectively.}
	\label{Fig:distribution}
\end{figure*}

Next, statistical tests are performed on binary outputs of the extraction algorithm $\mathbf{E}$. The parameter settings are the same as above, while the precision $z$ is set to $64$ and the parameter $K$ for the independence test is set to $10^3$.
With $3$ local maps, there are $3\times 3 =9$ pairs of $2$D CML instances as the input of $\mathbf{E}$. 
For each pair, one of the initial conditions of the CML instance is set to random and the other is obtained by perturbing the first one with difference up to $10^{-3}$. \modified{Note that the initial conditions are strongly correlated, so this is the worst-case study for the extraction algorithm $\mathbf{E}$. 
But, as argued in Sec.~\ref{subsec:randomness}, positive LE, on the other hand, can disperse the correlation within initial conditions. The resultant orbits (and thus PRD) are pseudo-independent of each other after a transition period.}
The binary outputs are collected and tested against NIST SP800-22 \cite{bassham2010sp} and TestU01 \cite{l2007testu01}.

\noindent \textbf{NIST SP800-22}:
{This is a statistical randomness test suite for binary sequences. It is a collection of $15$ tests, each of them outputs one (or, in a few cases, more than one) $P$-value. The suite is designed to be applied over binary sequences with length $10^6$. The significance level $\alpha$, which determines the region of acceptance of the assumption that the sequence being tested is random, is suggested to be set to $0.01$. A sequence passes a test if $P$-value is not less than $\alpha$.}

{Each test is designed such that the probability for an ideal generator to fail a test is $\alpha$.
So, $1000 > 1/ \alpha$ sequences with length $10^6$ produced by the extraction algorithm $\mathbf{E}$ are considered and tested. According to \cite[Chap.~4]{bassham2010sp}, empirical test results with $P$-value $\geq \alpha$ should be examined further according to\footnote{This is called the two-level random number testing procedure \cite{pareschi2012statistical}.}: 1) the pass rate of the sub-tests; 2) the distribution of empirical $P$-values. The pass rate is approximated with a normal random variable, and under the current setting, a pass rate $< 0.98056$ is interpreted as evidence of non-randomness. The $1000$ $P$-values from each sub-test will be further tested against a chi-square goodness-of-fit test with $10$ bins, and this again produces another $P\textnormal{-value}_T$ (i.e., a $P$-value of $P$-values). If $P\textnormal{-value}_T \geq 0.0001$, the sequences are considered to be uniform.}

\begin{table}[h]
\caption{NIST 800-22 Test Results on $\mathbf{E}$ with two Correlated L-CML Instances as Input.}
\label{Tab:NISTtest}
\centering
\resizebox{0.45\textwidth}{!}{%
\begin{tabular}{|l|l|l|l|}
\hline
\multicolumn{1}{|c|}{\multirow{2}{*}{Sub-test}} & $P$-value$^\dagger$     & Pass Rate     & $P$-value$_T$   \\ \cline{2-4} 
\multicolumn{1}{|c|}{}                          & $\geq 0.01$ & $\geq 0.98056$ & $\geq 0.0001$ \\ \hline
Frequency                                       & $0.8902$    & $0.9890$      & $0.8677$      \\ \hline
Block Freq.                                     & $0.7511$      & $0.9890$        & $0.9179$        \\ \hline
Cumulative*                                     & $0.3775$      & $0.9860$        & $0.1835$        \\ \hline
Runs                                            & $0.1881$      & $0.9900$        & $0.0242$        \\ \hline
Longest Run                                     & $0.9101$      & $0.9930$        & $0.6931$        \\ \hline
Rank                                            & $0.5842$      & $0.9900$        & $0.1381$        \\ \hline
Spectral (FFT)                                  & $0.6464$      & $0.9890$        & $0.5341$        \\ \hline
Nonoverlap*                                     & $0.4994$      & $0.9895$        & $0.4702$        \\ \hline
Universal                                       & $0.8006$      & $0.9820$        & $0.0669$        \\ \hline
App. Entropy                                    & $0.4894$      & $0.9830$        & $0.4245$        \\ \hline
Random E.*                                      & $0.3198$      & $0.9893$        & $0.5463$        \\ \hline
Random E. V.*                                   & $0.3176$      & $0.9915$        & $0.4760$        \\ \hline
Serial*                                         & $0.2898$      & $0.9880$        & $0.3429$        \\ \hline
Linear Comp.                                    & $0.4376$      & $0.9870$        & $0.1581$        \\ \hline
Success Counts                                  & $15/15$       & $15/15$         & $15/15$         \\ \hline
\multicolumn{4}{l}{* These tests generates more than one $P$-values; $\dagger$ Only} \\
\multicolumn{4}{l}{one of the test results has been considered in the table.} 
\end{tabular}%
}
\end{table}

\noindent \textbf{TestU01}: Generally speaking, it is a more systematic test suite for random numbers. TestU01 works for both random number over the interval $(0,1)$ and binary sequences, and it contains six predefined batteries of tests. In this work, the batteries \textit{Rabbit}, \textit{Alphabit} and \textit{BlockAlphabit} designed for binaries are selected. The batteries \textit{Rabbit} and \textit{Alphabit} have $38$ and $17$ statistical sub-tests, and \textit{BlockAlphabit} applies the sub-tests in \textit{Alphabit} repeatedly after reordering the bits by a block length specified by the software package. This test suite is run on binary sequences with length $2^{20}$, $2^{25}$ and $2^{30}$, and a $P\textnormal{-value}$ lies outside the range $[0.001, 0.999]$ is interpreted as fail; otherwise, it passes. 

Taking two correlated L-CML instances as input, $\mathbf{E}$ produces $10^3 \times 10^6$ bits and these bits are tested against NIST 800-22.
Table~\ref{Tab:NISTtest} shows the results of each sub-test of NIST 800-22, where the two-level testing procedure discussed above is applied. Note that the $P$-value displayed here is just a typical test output.
It is \modified{obvious} that the output of the extraction algorithm $\mathbf{E}$ passes \modified{all the sub-tests contained in the NIST 800-22 test suite.} 
Similarly, taking two correlated T-CML instances as input, the TestU01 test suite is run on the outputs and the results are summarized in Table~\ref{Tab:TestU01test}. Clearly, the output of the extraction algorithm $\mathbf{E}$ also passes the TestU01 test. \modified{In addition}, \modified{all additional $8$ pairs of $2$D CML instance combination of $\mathbf{E}$ pass both NIST 800-22 and TestU01 tests, but the results are omitted here owing to space constraints.} 
The experimental results confirm that the developed theory provides a sound foundation for generating pseudo-random bits from digital chaos and the extraction algorithm $\mathbf{E}$ can be used to produce uniform bits for secure applications.

\begin{table}[ht]
\caption{TestU01 Test Results on $\mathbf{E}$ with two Correlated T-CML Instances as Input.}
\label{Tab:TestU01test}
\centering
\resizebox{0.38\textwidth}{!}{%
\begin{tabular}{|c|c|c|c|}
\hline
\multicolumn{1}{|l|}{Length} & \multicolumn{1}{l|}{\textit{Rabbit}} & \multicolumn{1}{l|}{\textit{Alphabit}} & \multicolumn{1}{l|}{\textit{BlockAlphabit}} \\ \hline
$2^{30}$                     & $38/38$                     & $17/17$                       & $17/17$                            \\ \hline
$2^{25}$                     & $38/38$                     & $17/17$                       & $17/17$                            \\ \hline
$2^{20}$                     & $38/38$                     & $17/17$                       & $17/17$                            \\ \hline
\end{tabular}%
}
\end{table}

\modified{
\subsection{Efficiency Analysis}
To further assess the performance of the proposed extraction algorithm $\mathbf{E}$, it is evaluated through comparing with other chaotic PRNGs \cite{shujun2001pseudo,hua2015dynamic,alawida2020enhanced,lv2018novel}.
It is noted that the design in \cite{shujun2001pseudo} is the only previously known method that provides theoretically guaranteed uniform randomness from chaotic systems. 
By enhancing $1$D chaotic systems, the work in \cite{hua2015dynamic} is a famous heuristic PRNG design due to its simplicity and thorough experimental evaluation. 
The methods in \cite{alawida2020enhanced} and \cite{lv2018novel} are more recent heuristic proposals based on enhancing simple chaotic systems and using CML, respectively. 
For the proposed method, L-CML is used. For PRNGs in \cite{shujun2001pseudo,hua2015dynamic,alawida2020enhanced,lv2018novel}, the same settings of the original works are used. All the algorithms are then implemented on a Laptop with the Core i7-10710U CPU and 16G RAM.}

\modified{
Table \ref{Tab:Comparison} lists the running time (averaged from $1,000$ tests) for generating $1$ MByte binary stream from all these methods. With a running time of $40$ \textit{ms}, the proposed method is more efficient than the only other theoretical sound RPNG \cite{ shujun2001pseudo}, and is also more efficient than the heuristic designs \cite{alawida2020enhanced} (with running time $2996$ \textit{ms}) and \cite{lv2018novel} (with running time $58$ \textit{ms}), but inferior to the method in \cite{hua2015dynamic} (with running time $16$ \textit{ms}). 
However, looking further at the third and fourth column of Table \ref{Tab:Comparison}, it is clear that the proposed method provides theoretical randomness guarantee while the method in \cite{hua2015dynamic} does not.}

\modified{
\begin{table}[ht]
    \caption{\textcolor{black}{Running time comparison.}}
    \label{Tab:Comparison}
    \centering
    \resizebox{0.4\textwidth}{!}{%
        \begin{tabular}{|c|c|c|c|}
        \hline
        \multicolumn{1}{|c|}{Methods} & \multicolumn{1}{l|}{\shortstack{Running Time\\(1MByte)}}& \multicolumn{1}{l|}{\shortstack{Theoretical\\ Analysis}} & \multicolumn{1}{l|}{\shortstack{Experimental\\Analysis}} \\  \hline
         Ours & 40 \textit{ms}& Yes  & Yes  \\ \hline
        \cite{hua2015dynamic} & $16$ \textit{ms} & No & Yes  \\ \hline
        \cite{alawida2020enhanced} & $2996$ \textit{ms} & No & Yes   \\ \hline
        \cite{lv2018novel}& $58$ \textit{ms}  & No &  Yes  \\ \hline
         \cite{shujun2001pseudo} & $47$ \textit{ms} & Yes &  Yes  \\ \hline
        \end{tabular}%
    }
\end{table}
}

\modified{
To investigate the reason of the above experimental results, the amount of the basic operations for producing $8$ pseudo-random bits is used as the metric to evaluate the complexity of the considered PRNGs.
Referring the arithmetic details of the works in \cite{shujun2001pseudo,hua2015dynamic,alawida2020enhanced,lv2018novel}, the basic operations are counted and the result is listed in Table~\ref{Tab:basicoperations}. 
According to this table, the average number of the basic operations for the proposed method is $21$, which is smaller than that of \cite{shujun2001pseudo,alawida2020enhanced,lv2018novel} but bigger than $13.33$ of \cite{hua2015dynamic}. 
To summarize, the proposed method outperforms most of chaotic PRNGs \cite{shujun2001pseudo,hua2015dynamic,alawida2020enhanced,lv2018novel} in terms of efficiency. 
And the only method \cite{hua2015dynamic} that is more efficient than the proposed method does not provide guaranteed uniform randomness. From this sense, in real applications, one may choose to use the method in \cite{hua2015dynamic} when higher efficiency is needed, while opt to use the proposed method when higher security is desirable. 
}
\begin{table}[ht]
    \caption{\textcolor{black}{Number of basic operations for generating $8$ bits.}}
    \label{Tab:basicoperations}
    \centering
    \resizebox{0.45\textwidth}{!}{%
        \begin{tabular}{|l|l|l|l|l|l|}
        \hline
        \multicolumn{1}{|l|}{No. of Operations} &  \multicolumn{1}{|l|}{Ours} & \multicolumn{1}{|l|}{\cite{hua2015dynamic}} &  
        \multicolumn{1}{|l|}{\cite{alawida2020enhanced}}& \multicolumn{1}{|l|}{\cite{lv2018novel}} & \multicolumn{1}{|l|}{\cite{shujun2001pseudo}}\\  \hline
        No. of Exclusive OR     & $8$   & $0$   & $0$   & $0$ & $0$ \\ \hline
        No. of Interception     & $0$   & $2/3$ & $0$   & $1$  & $0$\\ \hline
        No. of Modulo           & $0$   & $8$   &  $8$  &  $0$ & $0$ \\ \hline
        No. of Compare          & $0$   & $1/3$   & $4$   & $0$ & $32$\\ \hline
        No. of Inversion        & $8$ & $0$ &$0$  &$0$  &$0$ \\ \hline
        No. of Addition/Subtraction     & $9/4$ & $8/3$ &  $248$ & $14$  & $8$  \\ \hline
        No. of Multiplication/Division  & $11/4$   & $5/3$ & $1368$ &  $26$ & $0$\\ \hline
        No. of Real $\to$ Char          & $0$  & $0$ & $0$ & $1$ & $0$ \\ \hline
        Total                           & $21$  & $13.33$ &  $1628$ & $42$ & $40$  \\ \hline
        \end{tabular}%
    }
\end{table}

\section{Conclusion}
\label{sec:conclusion}
Coupling chaotic maps is a \modified{popular method for generating more complicated} dynamic behavior, using for example $1$D and $2$D coupled map lattices in secure communications. This work presents the first theoretic study of the Lyapunov exponents of the $2$D CML model and finds that the maximum LE is solely determined by the local map used for coupling. Moreover, \modified{by bridging true randomness and pseudo-randomness, it lays the theoretic foundation for deriving uniform pseudo-randomness
from digitized chaotic orbits}. Making use of this result, a random number extraction algorithm $\mathbf{E}$ is designed, which produces pseudo-random bits with bias bounded by $O(2^{-z})$, where $z$ is the bit number of the precision. 
Extensive experiments are carried out and the results align perfectly with the theoretic formulas. \modified{Moreover, it is validated that the proposed algorithm possesses good efficiency. 
The theory may provide fresh insights into how chaos can be used to create pseudo-randomness.
Future work will include more performance gains through parallel implementation and theoretical analysis of pseudo-random distributions produced by chaos.
}

\appendix[Proof of Theorem~2 and Property~2]
\label{sec:appendix}
\begin{proof}
Let $P(x)$ be the density function of $x \in [0,1]$, which bounds the first derivative $P'(x)$. Considering that $x$ is represented with $z$ bits, one has $P(w_{z}=0) + P(w_{z}=1) =1$ and
\begin{IEEEeqnarray}{rCl}
 P(w_{z}=0)&=&\sum_{b=1}^{2^{z-1}} \int_\frac{2(b-1)}{2^z}^{\frac{2b-1}{2^z}} P(x) dx,
\label{Eq:expectation0}  \\
 P(w_{z}=1)&=&\sum_{b=1}^{2^{z-1}} \int_\frac{2b-1}{2^z}^{\frac{2b}{2^z}} P(x) dx.
\label{Eq:expectation1}
\end{IEEEeqnarray}
Applying the mean-value theorem to Eqs.~(\ref{Eq:expectation0}) and (\ref{Eq:expectation1}), one gets
\begin{IEEEeqnarray}{rCl}
 P(w_{z}=0)
 &=& \sum_{b=1}^{2^{z-1}}\int_{\frac{2(b-1)}{2^z}}^{\frac{2b-1}{2^z}} P(x) dx \nonumber \\ 
 &=& \int_{0}^\frac{1}{2^z}P(x)dx         +\int_\frac{2}{2^z}^\frac{3}{2^z}P(x)dx+\cdots \nonumber \\
 && 
 +\int_\frac{2^Z-2}{2^z}^{\frac{2^Z-1}{2^z}} P(x)dx \nonumber\\
 &=& \frac{1}{2^z} \left[P(x_1)+P(x_2)+\cdots+P(x_{2^{z-1}})\right], \nonumber
\end{IEEEeqnarray}
where $\frac{2(i-1)}{2^z}\le x_i \le \frac{2i-1}{2^z}$ for $i\in [1,2^{z-1}]$, and similarly,
\begin{IEEEeqnarray}{rCl}
 P(w_{z}=1)
 &=& \sum_{b=1}^{2^{z-1}}\int_{\frac{2b}{2^z}}^{\frac{2b-1}{2^z}} P(x) dx \nonumber \\ 
 &=& \frac{1}{2^z} \left[P(x'_1)+P(x'_2)+\cdots+P(x'_{2^{z-1}})\right], \nonumber
\end{IEEEeqnarray}
where $\frac{2i-1}{2^z}\le x'_i \le \frac{2i}{2^z}$ for $i\in [1,2^{z-1}]$. Finally, one has \begin{IEEEeqnarray}{rCl}
&& \lim_{z \to \infty} \left| P(w_{z}=1)-P(w_{z}=0) \right| \nonumber \\
& \leq & \lim_{z \to \infty} \frac{1}{2^z} \left( \sum_{i=1}^{2^{z-1}} \left| P(x_i) - P(x'_i)\right| \right) \nonumber \\
& \leq & \lim_{z \to \infty} \frac{1}{2^z} \cdot \left( \sum_{i=1}^{2^{z-1}} \left| P'(\bar{x}_i) \right| \cdot \frac{2}{2^z} \right) 
\label{Eq:meanvaluetheorem}\\
& \leq &  \lim_{z \to \infty} \frac{1}{2^z} \cdot \left( 2^{z-1}  \max{ \left| P'(\bar{x}_i) \right| } \right)  \cdot \frac{2}{2^z} \nonumber \\
& \leq &  \lim_{z \to \infty} \frac{1}{2^z} \cdot \left(  \max{ \left| P'(\bar{x}_i) \right| } \right)   \nonumber \\
& =& 0, \nonumber
\end{IEEEeqnarray}
where $\bar{x}_i \in (x_i, x'_i)$ and Eq.~(\ref{Eq:meanvaluetheorem}) is derived based on the mean value theorem.
\end{proof}

\begin{proof}
By assumption, one has
\begin{IEEEeqnarray}{rCL}
|P(x_i=0)- P(x_i=1)| &=& O(2^{-i}),  \nonumber \\
|P(y_j=0) - P(y_j=1)| &=& O(2^{-j}). \nonumber 
\end{IEEEeqnarray}
Since $w= x_i+y_j \pmod{2}$ and  $x_i$ and $y_j$ are independent of each other, one gets 
\begin{IEEEeqnarray}{rCL}
P(w=0)  &=& P(x_i=0)P(y_j=0) + P(x_i=1) P(y_j=1)\nonumber, \\
P(w=1)  &=& P(x_i=0)P(y_j=1) + P(x_i=1) P(y_j=0)\nonumber.
\end{IEEEeqnarray}
The bias of $w$ can then be calculated as 
\begin{IEEEeqnarray}{rCL}
&& |P(w=0) - P(w=1) |  \nonumber \\
&=& |P(x_i=0)P(y_j=0) + P(x_i=1) P(y_j=1) \nonumber \\ 
&&  -  P(x_i=0)P(y_j=1) - P(x_i=1) P(y_j=0)| \nonumber \\
&=& \left|
\left[P(x_i=0)- P(x_i=1)\right] \cdot 
\left[ P(y_j=0) - P(y_j=1)\right] \right| \nonumber \\
&=& O(2^{-(i+j)}). \nonumber
\end{IEEEeqnarray}
Hence, the property is true.
\end{proof}




%
\small{
\bibliographystyle{IEEEtran}
\bibliography{ref}
}

\end{document}